\documentclass[12pt]{amsart}
\usepackage{}

\usepackage{amsmath}
\usepackage{amsfonts}
\usepackage{amssymb}
\usepackage[all]{xy}           

\usepackage{bbding}
\usepackage{txfonts}
\usepackage{amscd}

\usepackage[shortlabels]{enumitem}
\usepackage{ifpdf}
\ifpdf
  \usepackage[colorlinks,final,backref=page,hyperindex]{hyperref}
\else
  \usepackage[colorlinks,final,backref=page,hyperindex]{hyperref}
\fi
\usepackage{tikz}
\usepackage[active]{srcltx}

\makeatletter

\newtheorem{thm}{Theorem}[section]
\newtheorem{lem}[thm]{Lemma}
\newtheorem{cor}[thm]{Corollary}
\newtheorem{pro}[thm]{Proposition}
\newtheorem{ex}[thm]{Example}
\newtheorem{rmk}[thm]{Remark}
\newtheorem{defi}[thm]{Definition}

\newcommand {\emptycomment}[1]{}

\newcommand{\lon }{\,\rightarrow\,}
\newcommand{\be }{\begin{equation}}
\newcommand{\ee }{\end{equation}}

\newcommand{\g}{\mathfrak g}

\newcommand{\Real}{\mathbb R}

\newcommand{\huaL}{\mathcal{L}}

\newcommand{\huaE}{\mathcal{E}}

\newcommand{\huaC}{{\mathcal{C}}}

\newcommand{\huaH}{\mathcal{H}}

\newcommand{\CWM}{C^{\infty}(M)}

\newcommand{\XM}{\frkX(M)}

\newcommand{\frka}{\mathfrak a}

\newcommand{\frkg}{\mathfrak g}

\newcommand{\frkD}{\mathfrak D}
\newcommand{\frkE}{\mathfrak E}

\newcommand{\frkX}{\mathfrak X}

\newcommand{\Courant}[1]{\left\llbracket  #1\right\rrbracket }

\newcommand{\Id}{\rm{Id}}

\newcommand{\br}[1]{   [ \cdot,    \cdot  ]   }

\newcommand{\dM}{\mathrm{d}}

\newcommand{\Hom}{\mathrm{Hom}}

\newcommand{\Nat}{\mathbb N}
\newcommand{\Der}{\mathrm{Der}}

\newcommand{\preF}{pre-$F$-algebroid}
\newcommand{\preFs}{pre-$F$-algebroids}

\newcommand{\pr}{\mathrm{pr}}

\newcommand{\Int}{\mathbb Z}

\begin{document}

\title[$F$-algebroids and deformation quantization via pre-Lie algebroids]{$F$-algebroids and deformation quantization via pre-Lie algebroids}

 \author{John Alexander Cruz Morales}
\address{Departmento de matem\'aticas, Universidad Nacional de Colombia, Bogot\'a \\
Max Planck Institute for Mathematics, Vivatsgasse 7, Bonn, Germany}
\email{jacruzmo@unal.edu.co\\
alexcruz@mpim-bonn.mpg.de}

\author{Jiefeng Liu}
\address{School of Mathematics and Statistics, Northeast Normal University, Changchun 130024, Jilin, China}
\email{liujf534@nenu.edu.cn}

\author{Yunhe Sheng}
\address{Department of Mathematics, Jilin University, Changchun 130012, Jilin, China}
\email{shengyh@jlu.edu.cn}


\begin{abstract}
In this paper, first we introduce a new approach to the notion of  $F$-algebroids, which  is a generalization of $F$-manifold algebras and $F$-manifolds, and show that  $F$-algebroids are the corresponding semi-classical limits of pre-Lie formal deformations of commutative associative algebroids. Then we use the deformation cohomology of pre-Lie algebroids to study pre-Lie infinitesimal deformations and extension of pre-Lie $n$-deformations to pre-Lie $(n+1)$-deformations of a commutative associative algebroid. Next we develop the theory of  Dubrovin's dualities of $F$-algebroids with eventual identities and use Nijenhuis operators on $F$-algebroids to construct new $F$-algebroids. Finally we introduce the notion of   pre-$F$-algebroids, which is a generalization of $F$-manifolds with compatible flat connections. Dubrovin's dualities of pre-$F$-algebroids with eventual identities, Nijenhuis operators on pre-$F$-algebroids and their applications to integral systems are discussed.
\end{abstract}




\maketitle

\tableofcontents

\allowdisplaybreaks


\section{Introduction}\label{sec:intr}

The concept of Frobenius manifolds was introduced by Dubrovin in \cite{Dub95} as a geometrical manifestation of the Witten-Dijkgraaf-Verlinde-Verlinde (WDVV) associativity equations in the $2$-dimensional topological field theories. Hertling and Manin weakened the conditions of a Frobenius manifold and introduced the notion of an $F$-manifold in \cite{HerMa}. Any Frobenius manifold has an underlying $F$-manifold structure. $F$-manifolds   appear in many fields of mathematics such as singularity theory \cite{Her02}, integrable systems \cite{ABLR,ALo2,ALo3,DS04,DS11,Lor,LPR11},  quantum K-theory \cite{LYP}, information geometry \cite{NoManin}, operad \cite{Merku} and so on.

The notion of a Lie algebroid was introduced by Pradines in 1967, which
is a generalization of Lie algebras and tangent bundles. Just as Lie algebras are the
infinitesimal objects of Lie groups, Lie algebroids are the infinitesimal objects of Lie groupoids. See \cite{General theory of Lie groupoid and Lie algebroid} for general
theory about Lie algebroids. Lie algebroids are now an active domain of research, with applications in various parts of mathematics, such as geometric mechanics, foliation theory, Poisson geometry, differential equations, singularity theory, operad and so on.  The notion of a pre-Lie algebroid (also called a left-symmetric algebroid or a Koszul-Vinberg algebroid) is a geometric generalization of a
pre-Lie algebra. Pre-Lie algebras arose from the study of convex homogeneous cones, affine manifolds and affine structures on Lie groups, deformation and cohomology theory of associative algebras and then appear in many fields in mathematics and mathematical physics. See the survey article \cite{Pre-lie algebra in geometry} for more details on pre-Lie algebras and \cite{lsb,LSBC,Boyom1,Boyom2} for more details on cohomology  and applications of pre-Lie algebroids. In \cite{Dot}, Dotsenko  showed that the graded object of the filtration of the operad encoding pre-Lie algebras is the operad encoding
$F$-manifold algebras, where the notion of an $F$-manifold algebra is the underlying algebraic structure of an $F$-manifold. In \cite{LSB},    the notion of  pre-Lie formal deformations of commutative associative algebras was introduced and it was shown that $F$-manifold algebras are the corresponding semi-classical limits. This result is parallel to that the semi-classical limit of an associative formal deformation of a commutative associative algebra is a Poisson algebra.

In this paper, we introduce the notion of   $F$-algebroids, which is a generalization of $F$-manifold algebras and $F$-manifolds. There is a slight difference
between this $F$-algebroid and the one introduced in \cite{CGT}.  We introduce the notion of  pre-Lie formal deformations of commutative associative algebroids and show that $F$-algebroids are the corresponding semi-classical limits. Viewing a commutative associative algebroid as a pre-Lie algebroid, we show that pre-Lie infinitesimal deformation and extension of pre-Lie $n$-deformations to pre-Lie $(n+1)$-deformations of a commutative associative algebroid are classified by the second and the third cohomology groups of the pre-Lie algebroid respectively.

$F$-manifolds with eventual identities were introduced by Manin in \cite{Manin1} and then were studied  systematically by David and Strachan in \cite{DS11}. In this paper, we generalize Dubrovin's dualities of $F$-manifolds with eventual identities to the case of $F$-algebroids. We introduce the notion of (pseudo-)eventual identities on $F$-algebroids and develop the theory of Dubrovin's dualities of $F$-algebroids with eventual identities. We introduce the notion of Nijenhuis operators on $F$-algebroids and use them to construct new $F$-algebroids. In particular, a pseudo-eventual identity naturally gives a Nijenhuis operator on an $F$-algebroid.

The notion of an $F$-manifold with a compatible flat connection was introduced by Manin in \cite{Manin1}. Applications of $F$-manifolds with compatible flat connections also appeared in Painlev\'e equations \cite{ALo1,ALo3,KMS,Lor} and integral systems \cite{ABLR,ALo2,KwM,LP,LPR11}. In this paper, we introduce the notion of  pre-$F$-algebroids, which is a generalization of $F$-manifolds with compatible flat connections. A pre-$F$-algebroid gives rise to an $F$-algebroid. We also study pre-$F$-algebroids with eventual identities and give a characterization of such eventual identities. Furthermore, The theory of Dubrovin's dualities of pre-$F$-algebroids with eventual identities were developed.  We introduce the notion of a Nijenhuis operator on a pre-$F$-algebroid, and show that a Nijenhuis operator gives rise to a deformed pre-$F$-algebroid.

Mirror symmetry, roughly speaking, is a duality between symplectic and complex geometry. The theory of Frobenius and $F$-manifolds plays an important role in this duality. We expect that the notion of $F$-algebroids might also be relevant in understanding the mirror phenomenon. In particular, the Dubrovin's dual of $F$-algebroids constructed in this paper should be related to the mirror construction along the way the Dubrovin's dual of Frobenius manifolds is related, at least in some situations, with mirror symmetry. More precisely the question is: Could we consider the construction of Dubrovin's dual of $F$-algebroids as a kind of mirror construction?  In order to answer the question above, we might need to add some extra structures to $F$-algebroids and include those structures in the construction of the Dubrovin's dual. This would allow us to give a comprehensible interpretation of our construction as a manifestation of a mirror phenomenon. We want to follow this line of thought in future works.

The paper is organized as follows. In Section \ref{sec:defi}, we introduce the notion of  $F$-algebroids and  give some constructions of $F$-algebroids including the action $F$-algebroids and direct product $F$-algebroids. In particular, we show that Poisson manifolds give rise to action $F$-algebroids naturally. In Section \ref{sec:deformation quantization}, we study pre-Lie formal deformations of a commutative associative algebroid, whose semi-classical limits are $F$-algebroids. We show that the equivalence classes of pre-Lie infinitesimal deformations of a commutative associative algebroid $A$ are classified by the second cohomology group in the deformation cohomology of $A$. Furthermore, we study  extensions of pre-Lie $n$-deformations to pre-Lie $(n+1)$-deformations of a commutative associative algebroid $A$ and show that a pre-Lie $n$-deformation can be extendable if and only if its obstruction class in the third cohomology group of the commutative associative algebroid $A$ is trivial. In Section \ref{sec:construction}, we first study Dubrovin's duality of $F$-algebroids with eventual identities. Then we use Nijenhuis operators on $F$-algebroids to construct deformed $F$-algebroids. In Section \ref{sec:pre-F-algebroid}, first we introduce the notion of a pre-$F$-algebroid, and show that a pre-$F$-algebroid gives rise to an $F$-algebroid. Then we study Dubrovin's duality of pre-$F$-algebroids with eventual identities. Finally, we introduce the notion of a Nijenhuis operator on a pre-$F$-algebroid, and show that a Nijenhuis operator on a pre-$F$-algebroid gives rise to a deformed pre-$F$-algebroid. Finally we give some applications to integrable systems.

\vspace{2mm}

\noindent
{\bf Acknowledgements.} This research was  supported by NSFC (11922110,11901501). The second author also supported by the National Key Research and Development Program of China (2021YFA1002000). The first author wants to thank Max Planck Institute for Mathematics where the last part of this work was written for its hospitality and financial support.




\section{$F$-algebroids}\label{sec:defi}

In this section, we introduce the notion of  $F$-algebroids, which are generalizations of $F$-manifolds and $F$-manifold algebras. We give some constructions of $F$-algebroids including the action $F$-algebroids and direct product $F$-algebroids.
\begin{defi}
An {\bf $F$-manifold algebra} is a triple $(\g,[-,-],\cdot)$, where $(\g,\cdot)$ is a commutative associative algebra and $(\g,[-,-])$ is a Lie algebra, such that for all $x,y,z,w\in \g$, the Hertling-Manin relation holds:
\begin{equation}\label{eq:HM relation}
P_{x\cdot y}(z,w)=x\cdot P_{y}(z,w)+y\cdot P_{x}(z,w),
\end{equation}
where $P_{x}(y,z)$ is defined by
\begin{equation}\label{eq:P}
 P_{x}(y,z)=[x,y\cdot z]-[x,y]\cdot z-y\cdot [x,z].
\end{equation}
\end{defi}


\begin{ex}{\rm
Any Poisson algebra is an $F$-manifold algebra.}
\end{ex}

\begin{defi}
An {\bf $F$-manifold} is a pair $(M,\bullet)$, where $M$ is a smooth manifold and $\bullet$ is a $C^\infty(M)$-bilinear,  commutative, associative multiplication on the tangent bundle $TM$  such that $(\frkX(M),[-,-]_{\frkX(M)},\bullet)$ is an $F$-manifold algebra, where $[-,-]_{\frkX(M)}$ is the Lie bracket of vector fields.
\end{defi}


The notion of   Lie algebroids was introduced
by Pradines in 1967, as a generalization of Lie algebras and
tangent bundles. See \cite{General theory of Lie groupoid and Lie
algebroid} for the general theory about Lie algebroids.
\begin{defi}
  A {\bf Lie algebroid} structure on a vector bundle $A\longrightarrow M$ is
a pair that consists of a Lie algebra structure $[-,-]_A$ on
the section space $\Gamma(A)$ and a vector bundle morphism
$a_A:A\longrightarrow TM$, called the anchor, such that the
following relation is satisfied:
$$~[X,fY]_A=f[X,Y]_A+a_A(X)(f)Y,\quad \forall~X,Y\in\Gamma(A),~f\in
\CWM.$$
\end{defi}
We usually denote a Lie algebroid by $(A,[-,-]_A,a_A)$, or $A$
if there is no confusion.

\begin{defi}
  A {\bf commutative associative algebroid} is a vector bundle
   $A$ over $M$ equipped with a $C^\infty(M)$-bilinear, commutative, associative multiplication $\cdot_A$ on the section space $\Gamma(A)$.

\end{defi}

We denote a commutative associative algebroid by $(A,\cdot_A)$.

In the following, we give the notion of $F$-algebroids, which are generalizations of $F$-manifold algebras and $F$-manifolds.

\begin{defi}
An {\bf $F$-algebroid}  is a vector bundle
   $A$ over $M$ equipped with a bilinear operation $\cdot_A:\Gamma(A)\times \Gamma(A)\rightarrow \Gamma(A)$, a skew-symmetric bilinear bracket $[-,-]_A:\Gamma(A)\times \Gamma(A)\rightarrow \Gamma(A)$, and a bundle map $a_A:A\rightarrow TM$, called the anchor, such that $(A,[-,-]_A,a_A)$ is a Lie algebroid, $(A,\cdot_A)$ is a commutative associative algebroid and $(\Gamma(A),[-,-]_A,\cdot_A)$ is an $F$-manifold algebra.

\end{defi}

We denote an $F$-algebroid by $(A,[-,-]_A,\cdot_A,a_A)$.
\emptycomment{\begin{defi}
An {\bf $F$-algebroid}  is a vector bundle
   $A$ over $M$ equipped with a $C^\infty(M)$-bilinear, commutative, associative multiplication $\cdot_A:\Gamma(A)\otimes \Gamma(A)\rightarrow \Gamma(A)$, a skew-symmetric bilinear bracket $[-,-]_A:\Gamma(A)\otimes \Gamma(A)\rightarrow \Gamma(A)$, and a bundle map $a_A:A\rightarrow TM$, called the anchor, such that $(A,[-,-]_A,a_A)$ is a Lie algebroid and $(\Gamma(A),[-,-]_A,\cdot_A)$ is an $F$-manifold algebra.
We denote an $F$-algebroid by $(A,[-,-]_A,\cdot_A,a_A)$.
\end{defi}}

\begin{rmk}
In \cite{CGT}, the authors had already defined an $F$-algebroid. There is a slight
difference between the above definition of an $F$-algebroid and that one. In \cite{CGT}, it is  assumed that the base manifold has an $F$-manifold structure $(M,\bullet)$. An $F$-algebroid defined in \cite{CGT} is a vector bundle
   $A$ over $M$ equipped with a bilinear operation $\cdot_A:\Gamma(A)\times \Gamma(A)\rightarrow \Gamma(A)$, a skew-symmetric bilinear bracket $[-,-]_A:\Gamma(A)\times \Gamma(A)\rightarrow \Gamma(A)$, and a bundle map $a_A:A\rightarrow TM$, such that $(A,[-,-]_A,a_A)$ is a Lie algebroid, $(A,\cdot_A)$ is a commutative associative algebroid, $(\Gamma(A),[-,-]_A,\cdot_A)$ is an $F$-manifold algebra and
   \begin{equation}\label{eq:F-algebroid-CGT}
   a_A(X\cdot_A Y)=a_A(X)\bullet a_A(Y),\quad\forall~X,Y\in\Gamma(A).
   \end{equation}
\end{rmk}

\begin{ex}
 Any $F$-manifold algebra is an $F$-algebroid over a point. Let $(M,\bullet)$ be an $F$-manifold. Then $(TM,[-,-]_{\frkX(M)},\bullet,\Id)$ is an $F$-algebroid.
\end{ex}

\begin{defi}
Let $(A,[-,-]_A,\cdot_A,a_A)$ and $(B,[-,-]_B,\cdot_B,a_B)$ be $F$-algebroids on $M$. A bundle map $\varphi:A\longrightarrow B$ is
called a {\bf homomorphism}  of $F$-algebroids, if for all $~X,Y\in\Gamma(A)$, the following
conditions are satisfied:
\begin{eqnarray*}
  \varphi(X\cdot_A Y)=\varphi(X)\cdot_B \varphi(Y),\quad
   \varphi([X, Y]_A)=[\varphi(X),\varphi(Y)]_B,\quad
   a_B\circ \varphi=a_A.
\end{eqnarray*}
\end{defi}

\begin{defi}
Let $(A,[-,-]_A,\cdot_A,a_A)$ be an $F$-algebroid. A section $e\in\Gamma(A)$ is  called the {\bf identity} if  $e\cdot_A X=X$ for all $X\in\Gamma(A)$. We denote an $F$-algebroid $(A,[-,-]_A,\cdot_A,a_A)$ with an identity $e$ by $(A,[-,-]_A,\cdot_A,e,a_A)$.
\end{defi}

\begin{pro}\label{pro:HM relation-tensor}
Let $(A,[-,-]_A,a_A)$ be a Lie algebroid equipped with a $C^\infty(M)$-bilinear, commutative, associative multiplication $\cdot_A:\Gamma(A)\times \Gamma(A)\rightarrow \Gamma(A)$. Define
\begin{equation}\label{eq:HM equation}
\Phi(X,Y,Z,W):=P_{X\cdot_A Y}(Z,W)-X\cdot_AP_Y(Z,W)-Y\cdot_AP_X(Z,W),\quad \forall~X,Y,Z,W\in\Gamma(A),
\end{equation}
where $P$ is given by \eqref{eq:P}. Then $\Phi$ is a tensor field  of type $(4,1)$ and
\begin{equation}
\Phi(X,Y,Z,W)=\Phi(Y,X,Z,W)=\Phi(X,Y,W,Z).
\end{equation}
\end{pro}
\begin{proof}
  By the  commutativity of the associative multiplication $\cdot_A$, we have
  $$\Phi(X,Y,Z,W)=\Phi(Y,X,Z,W)=\Phi(X,Y,W,Z).$$
  To prove that $\Phi$ is a tensor field  of type $(4,1)$, we only need to show
  $$\Phi(fX,Y,Z,W)=\Phi(X,Y,fZ,W)=f\Phi(X,Y,Z,W).$$
 By a direct calculation, we have
 \begin{eqnarray*}
 && \Phi(fX,Y,Z,W)\\
  &=&[f(X\cdot_A Y),Z\cdot_A W]_A-Z\cdot_A [f(X\cdot_A Y),W]_A-W\cdot_A [f(X\cdot_A Y),Z]_A\\
  &&-f\big(X\cdot_A P_Y(Z,W)\big)-Y\cdot_A \big([fX,Z\cdot_A W]_A-Z\cdot_A[fX,W]_A-W\cdot_A[fX,Z]_A\big)\\
  &=& fP_{X\cdot_A Y}(Z,W)-a_A(Z\cdot_A W)(f)(X\cdot_A Y)+a_A(W)(f)(X\cdot_A Y\cdot_A Z)\\
  &&+a_A(Z)(f)(X\cdot_A Y\cdot_A W)-f\big(X\cdot_A P_Y(Z,W)\big)-f\big(Y\cdot_A P_X(Z,W)\big)\\
  &&+a_A(Z\cdot_A W)(f)(X\cdot_A Y)-a_A(W)(f)(X\cdot_A Y\cdot_A Z)-a_A(Z)(f)(X\cdot_A Y\cdot_A W)\\
  &=& f \Phi(X,Y,Z,W).
 \end{eqnarray*}
 Similarly, we also have $\Phi(X,Y,fZ,W)=f\Phi(X,Y,Z,W)$.
\end{proof}

\begin{pro}\label{pro:identity}
 Let $(A,[-,-]_A,\cdot_A, a_A)$ be an $F$-algebroid with an identity $e$. Then $$P_e(X,Y)=0.$$
\end{pro}
\begin{proof}
It follows from  \eqref{eq:HM relation} directly.
\end{proof}

\begin{defi}
Let $(\frkg,[-,-],\cdot)$ be an $F$-manifold algebra. An {\bf action}
of $\frkg$ on a manifold $M$ is a linear map $\rho:\frkg\longrightarrow\frkX(M)$  from $\g$ to the space of vector fields on $M$,
such that for all $x,y\in\frkg$, we have
 $$
 \rho([x,y])=[\rho(x),\rho(y)]_{\frkX(M)}.
 $$
\end{defi}

 Given an action of
$\frkg$ on $M$, let $A=M\times
\g$ be the trivial bundle. Define an anchor map $a_\rho:A\longrightarrow TM$, a multiplication $\cdot_\rho:\Gamma(A)\times \Gamma(A)\longrightarrow \Gamma(A)$ and a bracket $[-,-]_\rho:\Gamma(A)\times \Gamma(A)\longrightarrow \Gamma(A)$   by
\begin{eqnarray}
a_\rho(m,u)&=&\rho(u)_m,\quad \forall ~m\in M, u\in\g,\label{action F1}\\
{X\cdot_\rho Y}&=&X\cdot Y, \label{action F2}\\
{[X,Y]_\rho}&=&\huaL_{\rho(X)}Y-\huaL_{\rho(Y)}X+[X,Y],\quad \forall~X,Y\in\Gamma(A), \label{action F3}
\end{eqnarray}
where  $X\cdot Y$ and $[X,Y]$ are the pointwise $C^{\infty}(M)$-bilinear multiplication and bracket, respectively.
\begin{pro}\label{pro:action F-algebroid}
With the above notations, $(A=M\times\frkg,[-,-]_\rho,\cdot_\rho,a_\rho)$ is an $F$-algebroid, which is called  an {\bf action $F$-algebroid},
where $[-,-]_\rho$, $\cdot_\rho $ and $a_\rho$ are given by $(\ref{action F3})$, $(\ref{action F2})$
and $(\ref{action F1})$, respectively.
\end{pro}
\begin{proof} Note that the multiplication $\cdot_\rho $ is a $C^\infty(M)$-bilinear, commutative and associative multiplication and $(A,[-,-]_\rho,a_\rho)$ is a Lie algebroid.
 By Proposition \ref{pro:HM relation-tensor} and the fact that $\g$ is an $F$-manifold algebra, for all $u_1,u_2,u_3,u_4\in\g$ and $f_1,f_2,f_3,f_4\in C^{\infty}(M)$, we have
 $$\Phi(f_1u_1,f_2u_2,f_3u_3,f_4u_4)=f_1f_2f_3f_4\Phi(u_1,u_2,u_3,u_4)=0,$$
which implies that $(\Gamma(A),[-,-]_\rho,\cdot_\rho)$ is an $F$-manifold algebra. Thus $(A,[-,-]_\rho,\cdot_\rho,a_\rho)$ is an $F$-algebroid.
\end{proof}

	\begin{ex}{\rm
  Let $\g$ be a $2$-dimensional vector space  with basis $\{e_1,e_2\}$. Then $(\g,[-,-],\cdot)$ with the non-zero multiplication $\cdot$ and the bracket $[-,-]$
  \begin{eqnarray*}
  e_1\cdot e_1=e_1,\quad e_1\cdot e_2=e_2\cdot e_1=e_2,\quad
  {[e_1,e_2]}= e_2
  \end{eqnarray*}
  is an $F$-manifold algebra with the identity $e_1$. Let $(t_1,t_2)$ be the canonical coordinate systems on $\Real^2$. It is straightforward to check that the map  $\rho:\g\longrightarrow \frkX(\Real^2)$ defined by
$$\rho(e_1)=t_2\frac{\partial}{\partial t_2},\quad \rho(e_2)=t_2\frac{\partial}{\partial t_1}+t^2_2\frac{\partial}{\partial t_2}$$
is an action  of the $F$-manifold algebra $\g$ on $\Real^2$. Then $(A=\Real^2\times\g,[-,-]_\rho,\cdot_\rho,a_\rho)$ is an $F$-algebroid with an identity $1\otimes e_1$,
where  $[-,-]_\rho$, $\cdot_\rho $ and $a_\rho$ are given by
\begin{eqnarray*}
a_\rho(m,c_1e_1+c_2 e_2)&=&(c_1t_2\frac{\partial}{\partial t_2}+c_2t_2\frac{\partial}{\partial t_1}+c_2t^2_2\frac{\partial}{\partial t_2})\mid_m,\quad \forall ~m\in \Real^2, \\
{f\otimes (c_1e_1)\cdot_\rho g\otimes (c_2e_i)}&=&(fg)\otimes (c_1c_2 e_i), \quad {f\otimes (c_1e_2)\cdot_\rho g\otimes (c_2e_2)}=0, \\
{[f\otimes (c_1e_1),g\otimes (c_2e_2)]_\rho}&=&fc_1t_2\frac{\partial g}{\partial t_2}\otimes (c_2 e_2)- gc_2(t_2\frac{\partial f}{\partial t_1}+t^2_2\frac{\partial f}{\partial t_2})\otimes (c_1 e_1)+ fg\otimes (c_1c_2[e_1, e_2]),
\end{eqnarray*}
where $f,g\in C^{\infty}(\Real^2),~c_1,c_2\in \Real,~i\in\{1,2\}$.

  }
\end{ex}

Let $(M,\pi)$ be a Poisson manifold and $(\CWM,\cdot,\{-,-\})$ be the corresponding Poisson algebra. Then for a given function $f$ on $M$, there is a unique vector field $H_f$ on $M$, called the Hamiltonian vector field of $f$, such that for any $g\in\CWM$, we have
$H_f(g)=\{f,g\}.$
Furthermore, the map $H:\CWM\to\XM$ defined by $ f\mapsto H_f$
 is a homomorphism from the Lie algebra $\CWM$ of smooth functions to the Lie algebra of smooth vector fields, i.e.
 $$H_{\{f,g\}}=[H_f,H_g]_{\XM},\quad \forall~f,g\in\CWM.$$

\begin{pro}
Let $(M,\pi)$ be a Poisson manifold and $(\CWM,\cdot,\{-,-\})$ be the corresponding Poisson algebra.
Then $(A=M\times\CWM,[-,-]_H,\cdot_H,a_H)$ is an $F$-algebroid, where the anchor map $a_H$, multiplication $\cdot_H$ and bracket $[-,-]_H$ are given   by
\begin{eqnarray*}
a_H(m,f)&=&H_f(m),\quad \forall ~m\in M, f\in\CWM,\label{Poisson action F1}\\
{(f_1 \otimes g_1)\cdot_H (f_2\otimes g_2)}&=&(f_1f_2)\otimes (g_1g_2), \label{Poisson action F2}\\
{[f_1 \otimes g_1,f_2\otimes g_2]_H}&=&f_1\{g_1,f_2\}\otimes \g_2-f_2\{g_2,f_1\}\otimes \g_1+f_1f_2\otimes \{g_1,g_2\}, \label{Poisson action F3}
\end{eqnarray*}
where $f_1,f_2,g_1,g_2$ are smooth functions on $M$.
\end{pro}

Let $A_1$ and $A_2$ be  vector bundles over $M_1$ and $M_2$ respectively. Denote the projections from $M_1\times M_2$ to $M_1$ and $M_2$ by $\pr_1$ and $ \pr_2$ respectively. The product vector bundle $A_1\times A_2\rightarrow M_1\times M_2$ can be regarded as the Whitney sum over $M_1\times M_2$ of the pullback vector bundles $\pr_1^!A_1$ and $\pr_2^!A_2$. Sections of $\pr_1^!A_1$ are of the form $\sum u_i\otimes X^1_i$, where $u_i\in C^\infty(M_1\times M_2)$ and $X^1_i\in\Gamma(A_1)$. Similarly, sections of $\pr_2^!A_2$ are of the form $\sum u'_i\otimes X^2_i$, where $u'_i\in C^\infty(M_1\times M_2)$ and $X^2_i\in\Gamma(A_2)$. The tangent bundle $T(M_1\times M_2)$ may in the same way be regarded as the Whitney sum $\pr_1^!(TM_1)\oplus \pr_2^!(TM_2)$. Let $(A_1,[-,-]_{A_1},a_{A_1})$ and $(A_2,[-,-]_{A_2},a_{A_2})$ be two Lie algebroids over the base manifolds $M_1$ and $M_2$ respectively. We define the anchor $\frka:A_1\times A_2\longrightarrow T(M_1\times M_2)$ by
\begin{equation*}\frka(\sum (u_i\otimes X^1_i)\oplus\sum (u'_j\otimes X^2_j) )=\sum (u_i\otimes a_{A_1}(X^1_i))\oplus\sum (u'_j\otimes a_{A_2}(X^2_j)).
\end{equation*}
And the Lie bracket on $A_1\times A_2$ is determined by the following relations with the Leibniz rule
\begin{eqnarray}
  \nonumber\Courant{1\otimes X^1,1\otimes Y^1}&=&1\otimes [X^1,Y^1]_{A_1},\quad \Courant{1\otimes X^1,1\otimes Y^2}=0,\\
  \nonumber\Courant{1\otimes X^2,1\otimes Y^2}&=&1\otimes [X^2,Y^2]_{A_2},\quad \Courant{1\otimes X^2,1\otimes Y^1}=0,
\end{eqnarray}
for $X^1,Y^1\in \Gamma(A_1),X^2,Y^2\in\Gamma(A_2)$. See \cite{General theory of Lie groupoid and Lie algebroid} for more details of the direct product Lie algebroids.

\begin{pro}
  Let $(A_1,[-,-]_{A_1},\cdot_{A_1},a_{A_1})$ and $(A_2,[-,-]_{A_2},\cdot_{A_2},a_{A_2})$ be two $F$-algebroids over the base manifolds $M_1$ and $M_2$ respectively. Then   $(A_1\times A_2,\Courant{-,-},\diamond,\frka)$ is  an $F$-algebroid over $M_1\times M_2$, where for
  $$X=\sum (u_i\otimes X^1_i)\oplus\sum (u'_j\otimes X^2_j),\quad Y=\sum (v_k\otimes Y^1_k)\oplus\sum (v'_l\otimes Y^2_l),$$
   the associative multiplication $\diamond$ is defined by
  \begin{eqnarray*}
    X\diamond Y=\sum (u_iv_k\otimes (X^1_i\cdot_{A_1}Y^1_k ))\oplus\sum (u'_jv'_l\otimes (X^2_j\cdot_{A_2}Y^2_l )).
  \end{eqnarray*}
\end{pro}

\begin{proof}
  It follows from straightforward verifications.
\end{proof}

The $F$-algebroid $(A_1\times A_2,\Courant{-,-},\diamond,\frka)$ is called the direct product $F$-algebroid.

\section{Pre-Lie deformation quantization of commutative associative algebroids}\label{sec:deformation quantization}

In this section, we study pre-Lie formal deformations of a commutative associative algebroid, whose semi-classical limits are $F$-algebroids. Viewing the commutative associative algebroid $A$ as a pre-Lie algebroid, we show that the equivalence classes of pre-Lie infinitesimal deformations of a commutative associative algebroid $A$ are classified by the second cohomology group in the deformation cohomology of $A$ and a pre-Lie $n$-deformation can be extended to a pre-Lie $(n+1)$-deformation if and only if its obstruction class in the third cohomology group of $A$ is trivial.

\begin{defi}
A {\bf pre-Lie algebra} is a pair $(\frkg,\ast)$, where $\g$ is a vector space  and  $\ast:\g\otimes \g\longrightarrow\g$ is a bilinear multiplication
satisfying that for all $x,y,z\in \g$, the associator
\begin{equation}\label{eq:associator}
(x,y,z)\triangleq x\ast(y\ast z)-(x\ast y)\ast z
\end{equation} is symmetric in $x,y$,
i.e.
$$(x,y,z)=(y,x,z),\;\;{\rm or}\;\;{\rm
equivalently,}\;\;x\ast(y\ast z)-(x\ast y)\ast z=y\ast(x\ast z)-(y\ast x)\ast
z.$$
\end{defi}

\begin{defi}{\rm(\cite{LSBC,Boyom1})}\label{defi:left-symmetric algebroid}
A {\bf pre-Lie algebroid} structure on a vector bundle
$A\longrightarrow M$ is a pair that consists of a pre-Lie
algebra structure $\ast_A$ on the section space $\Gamma(A)$ and a
vector bundle morphism $a_A:A\longrightarrow TM$, called the anchor,
such that for all $f\in\CWM$ and $X,Y\in\Gamma(A)$, the following
conditions are satisfied:
\begin{itemize}
\item[\rm(i)]$~X\ast_A(fY)=f(X\ast_A Y)+a_A(X)(f)Y$,
\item[\rm(ii)] $(fX)\ast_A Y=f(X\ast_A Y).$
\end{itemize}
\end{defi}

We usually denote a pre-Lie algebroid by $(A,\ast_A, a_A)$.
Any pre-Lie  algebra is a pre-Lie algebroid over a point.

  A connection $\nabla$ on a manifold $M$ is said to be {\bf flat} if the torsion and the curvature of the connection $\nabla$ vanish identically. A manifold $M$ endowed with a flat connection $\nabla$ is called a {\bf flat manifold}.

 \begin{ex}\label{ex:main}{\rm Let $M$
be a manifold with a flat connection
$\nabla$. Then $(TM,\nabla,\Id)$ is a pre-Lie algebroid whose sub-adjacent Lie algebroid is
exactly the tangent Lie algebroid. We denote this pre-Lie algebroid by $T_\nabla M$.
}
\end{ex}

\begin{pro}{\rm(\cite{LSBC})}\label{thm:sub-adjacent}
  Let $(A,\ast_A, a_A)$ be a pre-Lie algebroid. Define  a skew-symmetric bilinear bracket operation $[-,-]_A$ on $\Gamma(A)$ by
 \begin{equation}\label{eq:subadjacent bracket}
  [X,Y]_A=X\ast_A Y-Y\ast_A X,\quad \forall ~X,Y\in\Gamma(A).	
 \end{equation}
Then $(A,[-,-]_A,a_A)$ is a Lie algebroid, and denoted by
$A^c$, called the {\bf sub-adjacent Lie algebroid} of
 $(A,\ast_A,a_A)$.
\end{pro}

\begin{defi}
Let $E$ be a vector bundle over $M$. A {\bf multiderivation} of degree $n$ on $E$
is a pair $(D,\sigma_D)$, where  $D\in\Hom(\Lambda^{n-1}\Gamma(E)\otimes
\Gamma(E),\Gamma(E))$ and $\sigma_D\in \Gamma(\Hom(\Lambda^{n-1}E,TM))$, such that for all $f\in C^\infty(M)$ and sections
$X_i\in \Gamma(E)$, the following conditions are satisfied:
 \begin{eqnarray*}
D(X_1,\cdots,fX_i,\cdots,X_{n-1},X_{n})&=&f D(X_1,\cdots,X_i,\cdots,X_{n-1},X_{n}),\quad i=1,\cdots,n-1;\\
D(X_1,\cdots,X_{n-1},fX_{n})&=&f D(X_1,\cdots,X_{n-1},X_{n})+\sigma_D(X_1,\cdots,X_{n-1})(f)X_{n}.
\end{eqnarray*}
We will denote by $\Der^n(E)$ the space of multiderivations
of degree $n,~n\geq 1$.
\end{defi}

 Let $(A,\ast_A,a_A)$ be a pre-Lie algebroid. Recall that the deformation complex of $A$ is a cochain
complex $(\huaC_{\rm def}^{*}(A,A)=\bigoplus _{n\geq 0}\Der^n(A),\dM_{\rm def})$, where for all $X_i\in \Gamma(A),i=1,2\cdots,n+1$, the coboundary operator $\dM_{\rm def}:\Der^{n}(A)\longrightarrow \Der^{n+1}(A)$ is given
 by
 \begin{eqnarray*}
\dM_{\rm def}\omega(X_1,\cdots,X_{n+1})
 &=&\sum_{i=1}^{n}(-1)^{i+1}X_i\ast_A\omega(X_1,\cdots,\hat{X_i},\cdots,X_{n+1})\\
 &&+\sum_{i=1}^{n}(-1)^{i+1}\omega(X_1,\cdots,\hat{X_i},\cdots,X_n,X_i)\ast_A X_{n+1}\\
 &&-\sum_{i=1}^{n}(-1)^{i+1}\omega(X_1,\cdots,\hat{X_i},\cdots,X_n,X_i\ast_A X_{n+1})\\
 &&+\sum_{1\leq i<j\leq {n}}(-1)^{i+j}\omega([X_i,X_j]_A,X_1,\cdots,\hat{X_i},\cdots,\hat{X_j},\cdots,X_{n+1}),
\end{eqnarray*}
in which
$\sigma_{\dM_{\rm def}\omega}$
is given by
 \begin{eqnarray}
\nonumber\sigma_{\dM_{\rm def}\omega}(X_1,\cdots,X_n)
&=& \sum_{i=1}^{n}(-1)^{i+1}[a_A(X_i),\sigma_{\omega}(X_1,\cdots,\hat{X_i},\cdots,X_{n})]_{\frkX(M)}\\
\nonumber&&+\sum_{1\leq i<j\leq n}(-1)^{i+j}\sigma_{\omega}([X_i,X_j]_A,X_1,\cdots,\hat{X_i},\cdots,\hat{X_j},\cdots,X_n)\\
&&+\sum_{i=1}^{n}(-1)^{i+1}a_A(\omega(X_1,\cdots,\hat{X_i},\cdots,X_n,X_i)).\label{eq:simbol}
\end{eqnarray}
The corresponding cohomology, which we denote by $\huaH_{\rm def}^\bullet(A,A)$, is called the {\bf deformation cohomology} of the pre-Lie algebroid.

Since any commutative pre-Lie algebra is a commutative associative algebra, we have the following conclusion obviously.
\begin{lem}
 Any commutative pre-Lie algebroid is a commutative associative algebroid.
\end{lem}

Note that in a commutative pre-Lie algebroid, the anchor must be zero.

\begin{defi}
  Let $(A,\cdot_A)$ be a commutative associative algebroid. A {\bf pre-Lie formal deformation} of $A$ is a sequence of pairs $(\mu_k,\sigma_{\mu_k})\in\Der^2(A)$ with $\mu_0$ being the commutative associative algebroid multiplication $\cdot_A$ on $\Gamma(A)$ and $\sigma_{\mu_0}=0$ such that the $\Real[[\hbar]]$-bilinear product $\cdot_\hbar$ on $\Gamma(A)[[\hbar]]$ and $\Real[[\hbar]]$-linear map $\frka_\hbar:A\otimes\Real[[\hbar]]\rightarrow TM\otimes\Real[[\hbar]]$ determined by
\begin{eqnarray}
X\cdot_\hbar Y&=&\sum_{k=0}^\infty\hbar^k\mu_k(X,Y),\\
\frka_\hbar(X)&=&\sum_{k=0}^\infty\hbar^k \sigma_{\mu_k}(X),\quad\forall~X,Y\in \Gamma(A)
\end{eqnarray}
is a pre-Lie algebroid.
\end{defi}
It is straightforward to check that the rule of pre-Lie algebra product $\cdot_\hbar$ on $\Gamma(A)[[\hbar]]$ is equivalent to
  \begin{equation}\label{eq:pre-Lie rule}
  \sum_{i+j=k}\big(\mu_i(\mu_j(X,Y),Z)-\mu_i(X,\mu_j(Y,Z))\big)=\sum_{i+j=k}\big(\mu_i(\mu_j(Y,X),Z)-\mu_i(Y,\mu_j(X,Z))\big)
  \end{equation}
for $k\geq 0$.

\begin{thm}\label{thm:deformation quantization}
Let $(A,\cdot_A)$ be a commutative associative algebroid and $(A\otimes \Real[[\hbar]],\cdot_\hbar,\frka_\hbar)$ a pre-Lie formal deformation of $A$. Define a bracket $[-,-]_A:\Gamma(A)\times \Gamma(A)\longrightarrow \Gamma(A)$ by
$$[X,Y]_A=\mu_1(X,Y)-\mu_1(Y,X),\quad\forall~X,Y\in \Gamma(A).$$
Then $(A,[-,-]_A,\cdot_A,\sigma_{\mu_1})$ is an $F$-algebroid. The $F$-algebroid $(A,[-,-]_A,\cdot_A,\sigma_{\mu_1})$ is called the {\bf semi-classical limit} of $(A\otimes \Real[[\hbar]],\cdot_\hbar,\frka_\hbar)$. The pre-Lie algebroid $(A\otimes \Real[[\hbar]],\cdot_\hbar,\frka_\hbar)$ is called a {\bf pre-Lie deformation quantization} of $(A,\cdot_A)$.
\end{thm}
\begin{proof}
  Define the bracket $[-,-]_\hbar$ on $\Gamma(A)[[\hbar]]$ by
$$[X,Y]_\hbar=X\cdot_\hbar Y-Y\cdot_\hbar X=\hbar[X,Y]_A+\hbar^2(\mu_2(X,Y)-\mu_2(Y,X))+\cdots, \quad \forall~X,Y\in \Gamma(A).$$
By the fact that $(A\otimes \Real[[\hbar]],\cdot_\hbar,\frka_\hbar)$ is a pre-Lie algebroid, $(A[[\hbar]],[-,-]_\hbar,\frka_\hbar)$ is a Lie algebroid. The $\hbar^2$-terms of the Jacobi identity for $[-,-]_\hbar$ gives the Jacobi identity for $[-,-]_A$ and $\hbar$-terms of $[X,fY]_\hbar=f[X,Y]_\hbar+\frka_\hbar(X)(f) Y$ gives
 $$[X,fY]_A=f[X,Y]_A+\sigma_{\mu_1}(X)(f)Y.$$
  Thus $(A,[-,-]_A,\sigma_{\mu_1})$ is a Lie algebroid.

For $k=1$ in \eqref{eq:pre-Lie rule}, by the commutativity of $\mu_0$, we have
\begin{eqnarray*}
  &&\mu_0(\mu_1(X,Y),Z)-\mu_0(X,\mu_1(Y,Z))-\mu_1(X,\mu_0(Y,Z))\\
 & =&\mu_0(\mu_1(Y,X),Z)-\mu_0(Y,\mu_1(X,Z))-\mu_1(Y,\mu_0(X,Z)).
\end{eqnarray*}
By a similar proof given by Hertling in \cite{Her02}, we can show that the Hertling-Manin relation holds with $X\cdot_A  Y=\mu_0(X,Y)$ and $[X,Y]_A=\mu_1(X,Y)-\mu_1(Y,X)$ for $X,Y\in\Gamma(A)$.
Thus $(A,[-,-]_A,\cdot_A,\sigma_{\mu_1})$ is an $F$-algebroid.
\end{proof}

In the sequel, we study pre-Lie  $n$-deformations and pre-Lie  infinitesimal deformations of commutative associative algebroids.
\begin{defi}
  Let $(A,\cdot_A)$ be a commutative associative algebroid. A {\bf pre-Lie $n$-deformation} of $A$ is a sequence of pairs $(\mu_k,\sigma_{\mu_k})\in\Der^2(A)$ for $0\leq k\leq n$ with $\mu_0$ being the commutative associative algebroid multiplication $\cdot_A$ on $\Gamma(A)$ and $\sigma_{\mu_0}=0$, such that the $\Real[[\hbar]]/(\hbar^{n+1})$-bilinear product $\cdot_\hbar$ on $\Gamma(A)[[\hbar]]/(\hbar^{n+1})$ and $\Real[[\hbar]]/(\hbar^{n+1})$-linear map $\frka_\hbar:A\otimes\Real[[\hbar]]\rightarrow TM\otimes\Real[[\hbar]]$ determined by
\begin{eqnarray}
X\cdot_\hbar Y&=&\sum_{k=0}^n\hbar^k\mu_k(X,Y),\\
\frka_\hbar(X)&=&\sum_{k=0}^n\hbar^k \sigma_{\mu_k}(X),\quad\forall~X,Y\in \Gamma(A)
\end{eqnarray}
is a pre-Lie algebroid.
\end{defi}

We call a pre-Lie $1$-deformation of a commutative associative algebroid $(A,\cdot_A)$ a {\bf pre-Lie infinitesimal deformation} and denote it by $(A,\mu_1,a_A=\sigma_{\mu_1})$.

By direct calculations, $(A,\mu_1,\sigma_{\mu_1})$ is a pre-Lie infinitesimal deformation of a commutative associative algebroid $(A,\cdot_A)$  if and only if for all $X,Y,Z\in \Gamma(A)$
 \begin{eqnarray}
 \label{2-closed} &&\mu_1(X,Y)\cdot_A Z-X\cdot_A \mu_1(Y,Z)-\mu_1(X,Y\cdot_A Z)=\mu_1(Y,X)\cdot_A Z-Y\cdot_A \mu_1(X,Z)-\mu_1(Y,X\cdot_A Z).
\end{eqnarray}
Equation $(\ref{2-closed})$ means that  $\mu_1$ is a $2$-cocycle, i.e. $\dM_{{\rm def}}\mu_1=0$.

Two pre-Lie infinitesimal deformations  $A_\hbar=(A,\mu_1,\sigma_{\mu_1})$ and
$A'_\hbar=(A,\mu'_1,\sigma_{\mu'_1})$  of a commutative associative algebroid
$(A,\cdot_A)$ are said to be {\bf equivalent} if there exist
a family of pre-Lie algebroid homomorphisms
${\Id}+\hbar\varphi:A_\hbar\longrightarrow A'_\hbar$ modulo
$\hbar^2$ for $\varphi\in \Der^1(A)$. A pre-Lie infinitesimal deformation is said to be {\bf
trivial} if there exist  a family of pre-Lie algebroid
homomorphisms ${\Id}+\hbar\varphi:A_\hbar\longrightarrow
(A,\cdot_A,a_A=0)$ modulo $\hbar^2$.

By direct calculations, $A_\hbar$ and $A'_\hbar$  are
equivalent pre-Lie  infinitesimal deformations if and only if
\begin{eqnarray}
\sigma_{\mu_1}&=& \sigma_{\mu'_1},\label{eq:anchor}\\
\mu_1(X,Y)-\mu_1'(X,Y)&=&X\cdot_A \varphi(Y)+\varphi(X)\cdot_A Y-\varphi(X\cdot_A Y).\label{eq:2-exact}
\end{eqnarray}
Equation $(\ref{eq:2-exact})$ means that $\mu_1-\mu_1'=\dM_{\rm def}\varphi$ and \eqref{eq:anchor} can be obtained by \eqref{eq:2-exact}. Thus we have

\begin{thm}
 Let $(A,\cdot_A)$ be a commutative associative algebroid. There is a one-to-one correspondence between the space of equivalence classes of pre-Lie  infinitesimal deformations of $A$ and the second cohomology group $\huaH_{\rm def}^2(A,A)$.
\end{thm}

It is routine to check that

\begin{pro}
  Let $(A,\cdot_A)$ be a commutative associative algebroid such that $$\huaH_{\rm def}^2(A,A)=0.$$ Then all pre-Lie  infinitesimal deformations of $A$ are trivial.
\end{pro}

\begin{defi}
Let  $\{(\mu_1,\sigma_{\mu_1}), \cdots,(\mu_n,\sigma_{\mu_n})\}$   a pre-Lie $n$-deformation of a commutative associative algebroid $(A,\cdot_A)$. If there exists $(\mu_{n+1},\sigma_{\mu_{n+1}})\in\Der^2(A)$ such that  $$\{(\mu_1,\sigma_{\mu_1}), \cdots,(\mu_n,\sigma_{\mu_n}),(\mu_{n+1},\sigma_{\mu_{n+1}})\}$$ is a pre-Lie   $(n+1)$-deformation of $(A,\cdot_A)$, then  $\{(\mu_1,\sigma_{\mu_1}), \cdots,(\mu_n,\sigma_{\mu_n}),(\mu_{n+1},\sigma_{\mu_{n+1}})\}$ is   called an {\bf extension} of the pre-Lie   $n$-deformation   $\{(\mu_1,\sigma_{\mu_1}), \cdots,(\mu_n,\sigma_{\mu_n})\}$.
\end{defi}

\begin{thm}
  For any pre-Lie $n$-deformation of a  commutative associative algebroid $(A,\cdot_A)$, the pair $(\Theta_n,\sigma_{\Theta_n})\in \Der^3(A)$ defined by
\begin{eqnarray}
\label{eq:3-cocycle}\qquad\Theta_n(X,Y,Z)&=&\sum\limits_{i+j=n+1,\atop i,j\geq 1}\big(\mu_i(\mu_j(X,Y),Z)-\mu_i(X,\mu_j(Y,Z))-\mu_i(\mu_j(Y,X),Z)+\mu_i(Y,\mu_j(X,Z))\big),\\
\label{eq:3-cocycle symbol}\sigma_{\Theta_n}(X,Y)&=&\sum\limits_{i+j=n+1,\atop i,j\geq 1}\big(\sigma_{\mu_i}(\mu_j(X,Y)-\mu_j(Y,X))-[\sigma_{\mu_i}(X),\sigma_{\mu_j}(Y)]_{\frkX (M)}\big)
\end{eqnarray}
is a cocycle, i.e. $\dM_{\rm def}\Theta_n=0$.

 Moreover, the pre-Lie  $n$-deformation $\{(\mu_1,\sigma_{\mu_1}), \cdots,(\mu_n,\sigma_{\mu_n})\}$  extends to some pre-Lie  $(n + 1)$-deformation if and only if $[\Theta_n]=0$ in $\huaH_{\rm def}^3(A,A)$.
\end{thm}
\begin{proof}
It is obvious that
$
\Theta_n(X,Y,Z)=-\Theta_n(Y,Z,X)
$ for all $X,Y,Z\in \Gamma(A).$
It is straightforward to check that
\begin{eqnarray*}
\Theta_n(X,fY,Z)&=& f\Theta_n(X,Y,Z),\\
\Theta_n(X,Y,fZ)&=&f\Theta_n(X,Y,Z)+\sigma_{\Theta_n}(X,Y)(f)Z.
\end{eqnarray*}
Thus $\Theta_n$ is an element of $\Der^3(A)$. By a direct calculation,  the cochain $\Theta_n\in \Der^3(A)$ is closed.

Assume that the pre-Lie  $(n+1)$-deformation $\{(\mu_1,\sigma_{\mu_1}), \cdots,(\mu_{n+1},\sigma_{\mu_{n+1}})\}$ of a commutative associative algebroid $(A,\cdot_A)$    is an extension of the pre-Lie $n$-deformation     $\{(\mu_1,\sigma_{\mu_1}), \cdots,(\mu_n,\sigma_{\mu_n})\}$, then we have
\begin{eqnarray*}
  \label{eq:n+1 term}&&\sum_{i+j=n+1,i,j\geq 1}\big(\mu_i(\mu_j(X,Y),Z)-\mu_i(X,\mu_j(Y,Z))-\mu_i(\mu_j(Y,X),Z)+\mu_i(Y,\mu_j(X,Z))\big)\\
 \nonumber  &=&X\cdot_A \mu_{n+1}(Y,Z)-Y\cdot_A \mu_{n+1}(X,Z)+ \mu_{n+1}(Y,X)\cdot_A Z-\mu_{n+1}(X,Y)\cdot_A Z\\
   \nonumber&&+\mu_{n+1}(Y,X)\cdot_A Z-\mu_{n+1}(X,Y)\cdot_A Z.
\end{eqnarray*}
Note that the left-hand side of the above equality is just $\Theta_n(X,Y,Z)$. We can rewrite the above equality as
$$\Theta_n(X,Y,Z)=\dM_{\rm def}\mu_{n+1}(X,Y,Z).$$
We conclude that, if a pre-Lie $n$-deformation of a  commutative associative algebroid $(A,\cdot_A)$ extends to a pre-Lie  $(n + 1)$-deformation, then $\Theta_n$ is a coboundary.

Conversely, if $\Theta_n$ is a coboundary, then there exists an element $(\psi,\sigma_\psi )\in\Der^2(A)$ such that
$$\Theta_n(X,Y,Z)=\dM_{\rm def}\psi(X,Y,Z).$$
It is not hard to check that $\{(\mu_1,\sigma_{\mu_1}), \cdots,(\mu_{n+1},\sigma_{\mu_{n+1}})\}$ with $\mu_{n+1}=\psi$ is a  pre-Lie $(n+1)$-deformation of $(A,\cdot_A)$ and thus this pre-Lie  $(n+1)$-deformation is an extension of the pre-Lie $n$-deformation  $\{(\mu_1,\sigma_{\mu_1}), \cdots,(\mu_n,\sigma_{\mu_n})\}$.
\end{proof}

\section{Some constructions of $F$-algebroids}\label{sec:construction}

In this section, we use eventual identities and Nijenhuis operators to construct  $F$-algebroids. In particular, a pseudo-eventual identity naturally gives a Nijenhuis operator on an $F$-algebroid.
\subsection{(Pseudo-)Eventual identities and Dubrovin's dual of $F$-algebroids}
\begin{defi}
 Let $(A,[-,-]_A,\cdot_A,a_A)$ be an $F$-algebroid with an identity $e$. A section $\huaE\in\Gamma(A)$  is called a {\bf pseudo-eventual identity} on the $F$-algebroid if the following equality holds:
 \begin{equation}\label{eq:Eventual1}
  P_{\huaE}(X,Y)=[e,\huaE]_A\cdot_A X \cdot_A Y,\quad\forall~X,Y\in\Gamma(A).
 \end{equation}

 A pseudo-eventual identity $\huaE$ on the $F$-algebroid $A$ is called an {\bf eventual identity} if it is invertible, i.e. there is a section $\huaE^{-1}\in\Gamma(A)$ such that $\huaE^{-1}\cdot_A\huaE=\huaE\cdot_A \huaE^{-1}=e$.
\end{defi}
Denote the set of all pseudo-eventual identities on an $F$-algebroid $A$ by $E(A)$, i.e.
$$E(A)=\{\huaE\in\Gamma(A)\mid P_\huaE(X,Y)=[e,\huaE]_A\cdot_A X\cdot_A Y,\quad \forall~X,Y\in\Gamma(A)\}.$$

\begin{pro}\label{pro:property of eventual identity}
Let $(A,[-,-]_A,\cdot_A,a_A)$ be an $F$-algebroid with an identity $e$. Then $E(A)$ is an $F$-manifold subalgebra of $\Gamma(A)$. Moreover, if $\huaE\in\Gamma(A)$ is an eventual identity on the $F$-algebroid $A$, then $\huaE^{-1}$ is also an eventual identity on $A$.
\end{pro}
\begin{proof}
For the first claim, by a straightforward calculation, $E(A)$ is a subspace of the vector space $\Gamma(A)$.

Then we show that the multiplication $\cdot_A$ of two pseudo-eventual identities is still a pseudo-eventual identity. In fact, for any two  pseudo-eventual identities $\huaE_1$ and $\huaE_2$, by \eqref{eq:HM relation},
\begin{eqnarray*}
  P_{\huaE_1\cdot_A \huaE_2}(X,Y)&=&\huaE_1\cdot_A P_{\huaE_2}(X,Y)+\huaE_2\cdot_AP_{\huaE_1}(X,Y)\\
  &=&(\huaE_1\cdot_A [e,\huaE_2]_A+\huaE_2\cdot_A [e,\huaE_1]_A)\cdot_A X\cdot_A Y\\
  &=&[e,\huaE_1\cdot_A \huaE_2]_A\cdot_A X\cdot_A Y,
\end{eqnarray*}
where in the last equality we used $P_e(\huaE_1,\huaE_2)=0$. Thus $\huaE_1\cdot_A \huaE_2$ is a pseudo-eventual identity.

Finally, we show that the Lie bracket of two pseudo-eventual identities is also a pseudo-eventual identity. By \eqref{eq:HM relation}, we can show that for any $X,Y,Z,W\in\Gamma(A)$, we have
\begin{eqnarray*}
 P_{[X,Y]_A}(Z,W)&=&[X,P_Y(Z,W)]_A-P_Y([X,Z]_A,W)-P_Y(Z,[X,W]_A)\\
  &&-[Y,P_X(Z,W)]_A+P_X([Y,Z]_A,W)+P_X(Z,[Y,W]_A).
\end{eqnarray*}
Let $\huaE_1$ and $\huaE_2$ be eventual identities. Taking $X=\huaE_1$ and $Y=\huaE_2$, then by \eqref{eq:Eventual1},
\begin{eqnarray*}
 P_{[\huaE_2,\huaE_2]_A}(Z,W)&=&[\huaE_1,[e,\huaE_2]_A\cdot_A Z\cdot_A W]_A-[e,\huaE_2]_A\cdot_A[\huaE_1,Z]_A\cdot_A W\\
  &&-[e,\huaE_2]_A\cdot_A Z\cdot_A[\huaE_1,W]_A-[\huaE_2,[e,\huaE_1]_A\cdot_A Z\cdot_A W]_A\\
  &&+[e,\huaE_1]_A\cdot_A[\huaE_2,Z]_A\cdot_A W+[e,\huaE_1]_A\cdot_A Z\cdot_A[\huaE_2,W]_A.
\end{eqnarray*}
On the other hand, by \eqref{eq:Eventual1}, we have
\begin{eqnarray*}
 [\huaE_1,[e,\huaE_2]_A\cdot_A Z\cdot_A W]_A&=&2[e,\huaE_1]_A\cdot[e,\huaE_2]_A\cdot_AZ\cdot_A W+[\huaE_1,[e,\huaE_2]_A]_A\cdot_AZ\cdot_A W\\
 &&+[e,\huaE_2]_A\cdot_A[\huaE_1,Z]_A\cdot_A W+[e,\huaE_2]_A\cdot_A Z\cdot_A[\huaE_1,W]_A;\\
  {[\huaE_2,[e,\huaE_1]_A}\cdot_A Z\cdot_A W]_A&=&2[e,\huaE_2]_A\cdot[e,\huaE_1]_A\cdot_AZ\cdot_A W+[\huaE_2,[e,\huaE_1]_A]_A\cdot_AZ\cdot_A W\\
 &&+[e,\huaE_1]_A\cdot_A[\huaE_2,Z]_A\cdot_A W+[e,\huaE_1]_A\cdot_A Z\cdot_A[\huaE_2,W]_A.
\end{eqnarray*}
Thus
\begin{eqnarray*}
  P_{[\huaE_2,\huaE_2]_A}(Z,W)&=&[\huaE_1,[e,\huaE_2]_A]_A\cdot_AZ\cdot_A W-[\huaE_2,[e,\huaE_1]_A]_A\cdot_AZ\cdot_A W\\
&=&[e,[\huaE_1,\huaE_2]_A]_A\cdot_AZ\cdot_A W,
\end{eqnarray*}
which implies that $[\huaE_1,\huaE_2]_A$ is a pseudo-eventual identity. Therefore,  $E(A)$ is an  $F$-manifold subalgebra of $\Gamma(A)$.

Assume that $\huaE$ is an eventual identity on the $F$-algebroid $A$. By Proposition \ref{pro:identity}, we have $P_e(X,Y)=0$. Applying the Hertling-Manin relation with $X=\huaE$ and $Y=\huaE^{-1}$, we obtain
$$0=P_{\huaE\cdot_A \huaE^{-1}}(X,Y)=\huaE\cdot_AP_{\huaE^{-1}}(X,Y)+\huaE^{-1}\cdot_AP_{\huaE}(X,Y).$$
Combining this relation with \eqref{eq:Eventual1}, we have
$$P_{\huaE^{-1}}(X,Y)=-\huaE^{-2}\cdot_A[e,\huaE]_A\cdot_A X\cdot_A Y.$$
On the other hand, by  $P_e(X,Y)=0$, we have
$$0=P_e(\huaE,\huaE^{-1})=-[e,\huaE]_A\cdot_A \huaE^{-1}-\huaE\cdot_A[e,\huaE^{-1}]_A,$$
and then
$$[e,\huaE]_A\cdot_A\huaE^{-2}=([e,\huaE]_A\cdot_A\huaE^{-1})\cdot_A\huaE^{-1}=(-\huaE\cdot_A[e,\huaE^{-1}]_A)\cdot_A\huaE^{-1}=-[e,\huaE^{-1}]_A.$$
Thus we have
 \begin{equation*}\label{eq:Eventual2}
     P_{\huaE^{-1}}(X,Y)=[e,\huaE^{-1}]_A\cdot_A X \cdot_A Y,
 \end{equation*}
 which implies that $\huaE^{-1}$ is also an eventual identity on $A$.
\end{proof}

\emptycomment{\begin{pro}\label{pro:Eventual1}
 Let $(A,[-,-]_A,\cdot_A,a_A)$ be an $F$-algebroid with an identity $e$. If $\huaE\in\Gamma(A)$ is invertible, then
\eqref{eq:Eventual1} holds if and only if
 \begin{equation}\label{eq:Eventual2}
     P_{\huaE^{-1}}(X,Y)=[e,\huaE^{-1}]_A\cdot_A X \cdot_A Y,\quad\forall~X,Y\in\Gamma(A).
 \end{equation}
\end{pro}
\begin{proof} By Proposition \ref{pro:identity}, we have $P_e(X,Y)=0$. Applying the Hertling-Manin relation with $X=\huaE$ and $Y=\huaE^{-1}$, we obtain
$$0=P_{\huaE\cdot_A \huaE^{-1}}(X,Y)=\huaE\cdot_AP_{\huaE^{-1}}(X,Y)+\huaE^{-1}\cdot_AP_{\huaE}(X,Y).$$
Combining this relation with \eqref{eq:Eventual1}, we have
$$P_{\huaE^{-1}}(X,Y)=-\huaE^{-2}\cdot_A[e,\huaE]_A\cdot_A X\cdot_A Y.$$
On the other hand, by  $P_e(X,Y)=0$, we have
$$0=P_e(\huaE,\huaE^{-1})=-[e,\huaE]_A\cdot_A \huaE^{-1}-\huaE\cdot_A[e,\huaE^{-1}]_A,$$
and thus
$$[e,\huaE]_A\cdot_A\huaE^{-2}=([e,\huaE]_A\cdot_A\huaE^{-1})\cdot_A\huaE^{-1}=(-\huaE\cdot_A[e,\huaE^{-1}]_A)\cdot_A\huaE^{-1}=-[e,\huaE^{-1}]_A.$$
\eqref{eq:Eventual2} follows immediately. The converse can be proved similarly. \end{proof}}

A pseudo-eventual identity on an $F$-algebroid gives a new $F$-algebroid.

\begin{thm}\label{thm:construction F-algebroid}
  Let $(A,[-,-]_A,\cdot_A,a_A)$ be an $F$-algebroid with an identity $e$. Then $\huaE$ is a pseudo-eventual identity on $A$ if and only if $(A,[-,-]_A,\cdot_{\huaE},a_A)$ is  an $F$-algebroid, where $\cdot_{\huaE}:\Gamma(A)\times \Gamma(A)\longrightarrow \Gamma(A)$ is defined by
  \begin{equation}\label{eq:New associative mult}
    X\cdot_{\huaE} Y=X\cdot_A Y\cdot_A \huaE,\quad\forall~X,Y\in\Gamma(A).
  \end{equation}
\end{thm}
\begin{proof}
The proof of this theorem is similar to the proof of Theorem 3 in \cite{DS11}. We give a sketchy proof here for completeness. Assume that $\huaE$ is a pseudo-eventual identity on $A$. It is straightforward to check that the multiplication $\cdot_{\huaE}$ defined by \eqref{eq:New associative mult} is $\CWM$-bilinear, commutative and associative.

For $X,Y,Z\in\Gamma(A)$, we set
$$P^\huaE_{X}(Y,Z):=[X,Y\cdot_\huaE Z]_A-[X,Y]_A\cdot_\huaE Z-Y\cdot_\huaE [X,Z]_A.$$
By a direct calculation, we have
\begin{eqnarray}\label{eq:eventual relation 1}
P^\huaE_X(Y,Z)=P_X(\huaE\cdot_A Y,Z)+P_X(\huaE,Y)\cdot_A Z+[Z,\huaE]_A\cdot_A X\cdot_A Y.
\end{eqnarray}
Using \eqref{eq:eventual relation 1} and the Hertling-Manin relation, we have
\begin{eqnarray*}
P^\huaE_{X\cdot_\huaE Y}(Z,W)&=&\huaE\cdot_A X\cdot_AP_Y(\huaE\cdot_A Z,W)	+\huaE\cdot_A Y\cdot_AP_X(\huaE\cdot_A Z,W)\\
&&+X\cdot_A Y\cdot_AP_\huaE(\huaE\cdot_A Z,W)+\huaE\cdot_A X\cdot_AW\cdot_AP_Y(\huaE,Z)	\\
&&+\huaE\cdot_A Y\cdot_AW\cdot_AP_X(\huaE,Z)	+X\cdot_A Y\cdot_AW\cdot_AP_\huaE(\huaE,Z)	\\
&&+[X\cdot_A Y\cdot_A \huaE,\huaE]_A\cdot_A Z\cdot_A W.
\end{eqnarray*}
Since $\huaE$ is a pseudo-eventual identity on $A$,  by \eqref{eq:eventual relation 1}, we have
\begin{eqnarray*}
&&P^\huaE_{X\cdot_\huaE Y}(Z,W)-X\cdot_\huaE 	P^\huaE_{Y}(Z,W)-Y\cdot_\huaE 	P^\huaE_{X}(Z,W)\\
&=&X\cdot_A Y\cdot_A\big(P_\huaE(\huaE\cdot_A Z,W)+W\cdot_AP_\huaE(\huaE,Z)	\big)\\
&&-Z\cdot_A W\cdot_A\big([X\cdot_A Y\cdot_A \huaE,\huaE]_A+\huaE\cdot_A X\cdot_A[Y,\huaE]_A+\huaE\cdot_A Y\cdot_A[X,\huaE]_A\big)\\
&=&X\cdot_A Y\cdot_A\big(P_\huaE(\huaE\cdot_A Z,W)+W\cdot_AP_\huaE(\huaE,Z)	\big)-Z\cdot_A W\cdot_A\big(P_\huaE(\huaE,X)\cdot_A Y+P_\huaE(\huaE\cdot_A X,Y)\big).\\
&=&X\cdot_A Y\cdot_A([e,\huaE]_A\cdot_A\huaE\cdot_A Z\cdot_A W+[e,\huaE]_A\cdot_A\huaE\cdot_A Z\cdot_A W)\\
&&-Z\cdot_A W\cdot_A\big([e,\huaE]_A\cdot_A\huaE\cdot_A X\cdot_A Y+[e,\huaE]_A\cdot_A\huaE\cdot_A X\cdot_A Y\big)\\
&=&2[e,\huaE]_A\cdot_A\huaE\cdot_A X\cdot_A Y\cdot_A Z\cdot_A W-2[e,\huaE]_A\cdot_A\huaE\cdot_A X\cdot_A Y\cdot_A Z\cdot_A W\\
&=&0,
\end{eqnarray*}
 which implies
 that  $(A,[-,-]_A,\cdot_{\huaE},a_A)$ is an $F$-algebroid.	

 The converse can be proved similarly. We omit the details.
\end{proof}

\begin{thm}\label{thm:involution1}
  Let $(A,[-,-]_A,\cdot_A,a_A)$ be an $F$-algebroid with an identity $e$. Then $\huaE$ is an eventual identity on $A$ if and only if $(A,[-,-]_A,\cdot_{\huaE},a_A)$ is also an $F$-algebroid with the identity $\huaE^{-1}$, which is called the Dubrovin's dual of $(A,[-,-]_A,\cdot_A,a_A)$, where $\cdot_{\huaE}$ is given by \eqref{eq:New associative mult}.  Moreover, $e$ is an eventual identity on the $F$-algebroid $(A,[-,-]_A,\cdot_{\huaE},\huaE^{-1},a_A)$  and the map
  \begin{equation}\label{eq:involution}
  (A,[-,-]_A,\cdot_A,e,a_A,\huaE)\longrightarrow(A,[-,-]_A,\cdot_{\huaE},\huaE^{-1},a_A,e^{\dag})
  \end{equation}
  is an involution of the set of $F$-algebroids with eventual identities, where $e^{\dag}:=\huaE^{-2}$ is the inverse of $e$ with respect to the multiplication $\cdot_\huaE$.
\end{thm}
\begin{proof}
By Theorem \ref{thm:construction F-algebroid}, $(A,[-,-]_A,\cdot_{\huaE},a_A)$ is an $F$-algebroid. It is obvious that $\huaE^{-1}$ is the identity with respect to the multiplication $\cdot_{\huaE}$ defined by \eqref{eq:New associative mult}.

Next, we show that $e$ is an eventual identity on $(A,[-,-]_A,\cdot_{\huaE},\huaE^{-1},a_A)$. Since the identity with respective to the multiplication $\cdot_\huaE$ is $\huaE^{-1}$, we need to show that
$$[e,X\cdot_{\huaE} Y]_A-[e,X]_A\cdot_\huaE Y-X\cdot_\huaE [e,Y]_A=[\huaE^{-1}, e]_A\cdot_\huaE X \cdot_\huaE Y,\quad\forall~X,Y\in\Gamma(A).$$
By a straightforward computation, for any $Z\in\Gamma(A)$, we have
\begin{eqnarray}\label{eq:Eventua2}
\qquad ~[Z,X\cdot_{\huaE} Y]_A-[Z,X]_A\cdot_\huaE Y-X\cdot_\huaE [Z,Y]_A
  &=&P_Z(\huaE\cdot_A X,Y)+P_Z(\huaE,X)\cdot_A Y\\
 \nonumber &&+[Z,\huaE]_A\cdot_A X\cdot_A Y.
\end{eqnarray}
Letting $Z=e$ in \eqref{eq:Eventua2} and using $P_e(X,Y)=0$, we have
$$[e,X\cdot_{\huaE} Y]_A-[e,X]_A\cdot_\huaE Y-X\cdot_\huaE [e,Y]_A=[e,\huaE]_A\cdot_A X\cdot_A Y=([e,\huaE]_A\cdot_A \huaE^{-2})\cdot_\huaE X\cdot_\huaE Y.$$
Recall now from the proof of Proposition \ref{pro:property of eventual identity} that $[e,\huaE]_A\cdot_A \huaE^{-2}=[\huaE^{-1},e]_A$. Thus $e$ is an eventual identity on the $F$-algebroid $(A,[-,-]_A,\cdot_{\huaE},\huaE^{-1},a_A)$.

Now we show that the map \eqref{eq:involution} is an involution. Note that $e^{\dag}:=\huaE^{-2}$ is the inverse of $e$ with respect to the multiplication $\cdot_\huaE$. By Proposition \ref{pro:property of eventual identity}, $e^{\dag}$ is also an eventual identity on the $F$-algebroid $(A,[-,-]_A,\cdot_{\huaE},\huaE^{-1},a_A)$. Furthermore, for $X,Y\in \Gamma(A)$, we have
$$X\cdot_A Y=X\cdot_\huaE Y\cdot_\huaE \huaE^{-2}=X\cdot_\huaE Y\cdot_\huaE {e}^{\dag},$$
which implies that the map defined by \eqref{eq:involution} is an involution of the set of $F$-algebroids with eventual identities.
\end{proof}

\begin{defi}
  An $F$-manifold $(M,\bullet,e)$ is called {\bf semi-simple} if there exists canonical local coordinates $(u^1,\cdots,u^n)$ on $M$ such that $e=\frac{\partial}{\partial u^1}+\cdots+\frac{\partial}{\partial u^n}$ and
$$\frac{\partial}{\partial u^i}\bullet \frac{\partial}{\partial u^j}=\delta_{ij}\frac{\partial}{\partial u^j},\quad i,j\in\{1,2,\cdots,n\}$$
\end{defi}

\begin{ex}
Let $(M,\bullet,e)$ be a semi-simple  $F$-manifold. Then $(TM,[-,-]_{\frkX(M)},\bullet,\Id)$ is an $F$-algebroid with an identity $e$.
It is straightforward to check that any pseudo-eventual identity on  $(TM,[-,-]_{\frkX(M)},\bullet,\Id)$ is of the form
$$\huaE=f_1(u^1)\frac{\partial}{\partial u^1}+\cdots+f_n(u^n)\frac{\partial}{\partial x_n},$$
where $f_i(u^i)\in\CWM$ depends only on $u^i$ for $i=1,2,\cdots,n$. Furthermore, it was shown in \cite{DS11} that if all $f_i(u^i)$ are  non-vanishing everywhere, then $\huaE\in \XM$ is an  eventual identity.
\end{ex}

\subsection{Nijenhuis operators and deformed $F$-algebroids}

Recall from \cite{CGM} that a Nijenhuis operator on a commutative associative algebra $(A,\cdot_A)$ is a linear map $N:A\lon A$ such that
\begin{equation}\label{eq:Nij-Struc1}
 N(x)\cdot_A N(y)=N\big(N(x)\cdot_A y+x\cdot_A N(y)-N(x\cdot_A y)\big),\quad\forall~x,y\in A.
\end{equation}
Recall that a Nijenhuis operator on a Lie algebroid $(A,[-,-]_A,a_A)$ is a bundle map $N:A\lon A$ such that
\begin{equation}\label{eq:Nij-Struc3}
 [N(X),N(Y)]_A=N\big([N(X), Y]_A+[X,N(Y)]_A-N([X, Y]_A)\big),\quad\forall~X,Y\in\Gamma(A).
\end{equation}

Based on these structures, we introduce the notion of a Nijenhuis operator on an $F$-algebroid as follows.

\begin{defi}
  Let $(A,[-,-]_A,\cdot_A,a_A)$ be an $F$-algebroid. A bundle map $N:A\lon A$ is called a {\bf Nijenhuis operator} on the $F$-algebroid  $A$ if $N$ is both a Nijenhuis operator on the commutative associative algebra $(\Gamma(A),\cdot_A)$ and a Nijenhuis operator on the Lie algebroid $(A,[-,-]_A,a_A)$.
\end{defi}

Define the deformed operation $\cdot_N:\Gamma(A)\times\Gamma(A)\longrightarrow \Gamma(A)$ and  the deformed bracket $[-,-]_N:\Gamma(A)\times\Gamma(A)\longrightarrow \Gamma(A)$ by
\begin{eqnarray}
\label{eq:deformCA}  X\cdot_N Y&=&N(X)\cdot_A Y+X\cdot_A N(Y)-N(X\cdot_A Y),\\
\label{eq:deforLA} [X,Y]_N&=&[N(X),Y]_A+[X,N(Y)]_A-N([X,Y]_A),\quad\forall~X,Y\in \Gamma(A).
\end{eqnarray}
\begin{thm}\label{pro:Nijenhuis operator}
  Let $N:A\longrightarrow A$ be a Nijenhuis operator on an $F$-algebroid $(A,[-,-]_A,\cdot_A,a_A)$. Then $(A,[-,-]_N,\cdot_N,a_N=a_A\circ N)$
 is an $F$-algebroid and $N$ is an $F$-algebroid homomorphism from the $F$-algebroid $(A,[-,-]_N,\cdot_N,a_N=a_A\circ N)$ to $(A,[-,-]_A,\cdot_A,a_A)$.
\end{thm}
\begin{proof}
  Since $N$ is a Nijenhuis operator on the commutative associative algebra $(\Gamma(A),\cdot_A)$, it follows that $(\Gamma(A),\cdot_N)$ is a commutative associative algebra (\cite{CGM}). Since $N$ is a Nijenhuis operator on  the Lie algebroid $(A,[-,-]_A,a_A)$, $(A,[-,-]_N,a_N)$ is a Lie algebroid (\cite{Kosmann1}).

Define
\begin{equation}\label{eq:HM deformed equation}
\Phi_N(X,Y,Z,W):=P^N_{X\cdot_N Y}(Z,W)-X\cdot_NP^N_Y(Z,W)-Y\cdot_NP^N_X(Z,W),
\end{equation}
where $X,Y,Z,W\in\Gamma(A)$ and
$$P^N_{X}(Y,Z):=[X,Y\cdot_N Z]_N-[X,Y]_N\cdot_N Z-Y\cdot_N [X,Z]_N.$$
  Since $A$ is an $F$-algebroid and $N$ is a Nijenhuis operator on $A$, by a direct calculation, we have
  \begin{eqnarray*}
  	\Phi_N(X,Y,Z,W)
  	&=&\Phi(N(X), N(Y),N(Z),W)+\Phi(N(X), N(Y),Z,N(W)) \\
  	&& +\Phi(N(X), Y,N(Z),N(W))+\Phi(X, N(Y),N(Z),N(W))\\
  	&&-N\Big(\Phi(N(X), N(Y),Z,W)+\Phi(N(X), Y,N(Z),W)+\Phi(N(X), Y,Z,N(W))\\
  	&&+\Phi(X, N(Y),N(Z),W)+\Phi(X,Y,N(Z),N(W))+\Phi(X, N(Y),Z,N(W))\Big)\\&&+N^2\Big(\Phi(N(X), Y,Z,W)+\Phi(X, N(Y),Z,W)\\
  	&&+\Phi(X, Y,N(Z),W)+\Phi(X, Y,Z,N(W))\Big)\\
  &&-N^3(\Phi(X, Y,Z,W))\\
  &=&0,
  	\end{eqnarray*}
which implies that
$$P^N_{X\cdot_N Y}(W,Z)-X\cdot_NP^N_Y(W,Z)-Y\cdot_NP^N_X(W,Z)=0.$$
Thus $(A,[-,-]_N,\cdot_N,a_N=a_A\circ N)$
 is an $F$-algebroid. It is obvious that $N$ is an $F$-algebroid homomorphism from the $F$-algebroid $(A,[-,-]_N,\cdot_N,a_N=a_A\circ N)$ to $(A,[-,-]_A,\cdot_A,a_A)$.
\end{proof}

\begin{lem}\label{lem:Niejproperty}
  Let $(A,[-,-]_A,\cdot_A,a_A)$  be an $F$-algebroid and $N$ a Nijenhuis operator on $A$. For all $k,l\in\Nat$,
  \begin{itemize}
\item[$\rm(i)$]$(A,[-,-]_{N^k},\cdot_{N^k},a_{N^k})$ is an $F$-algebroid;
\item[$\rm(ii)$]$N^l$ is also a Nijenhuis operator on the $F$-algebroid $(A,[-,-]_{N^k},\cdot_{N^k},a_{N^k})$;
\item[$\rm(iii)$]The $F$-algebroids $(A,([-,-]_{N^k})_{N^l},(\cdot_{N^k})_{N^l},a_{N^{k+l}})$ and $(A,[-,-]_{N^{k+l}},\cdot_{N^{k+l}},a_{N^{k+l}})$ are the same;
\item[$\rm(iv)$]$N^l$ is an $F$-algebroid homomorphism between the $F$-algebroid $(A,[-,-]_{N^{k+l}},\cdot_{N^{k+l}},a_{N^{k+l}})$ and $(A,[-,-]_{N^k},\cdot_{N^k},a_{N^{k}})$.
  \end{itemize}
\end{lem}
\begin{proof}
  Since the above conclusions with respect to Nijenhuis operators on commutative associative algebras (\cite{CGM}) and Lie algebroids (\cite{Kosmann1}) simultaneously hold, by Theorem \ref{pro:Nijenhuis operator}, the conclusions follow immediately.
\end{proof}
At the end of this section, we show that a pseudo-eventual identity naturally gives a Nijenhuis operator on an $F$-algebroid.

\begin{pro}
Let $(A,[-,-]_A,\cdot_A,a_A)$ be an $F$-algebroid with an identity $e$ and	$\huaE$   a pseudo-eventual identity on $A$. Then the endomorphism $N=\huaE\cdot_A$ is a Nijenhuis operator on the $F$-algebroid $A$.
Consequently, $(A,[-,-]_\huaE,\cdot_\huaE,a_\huaE)$ is an $F$-algebroid, where the bracket $[-,-]_\huaE$ is given by
\begin{equation}
[X,Y]_\huaE=[\huaE\cdot_A X,Y]_A+[X,\huaE\cdot_A Y]_A-\huaE\cdot_A[X,Y]_A,\quad \forall~X,Y\in\Gamma(A),\end{equation}
the multiplication $\cdot_\huaE$ is given by \eqref{eq:New associative mult} and $a_\huaE(X)=a_A(\huaE\cdot_A X)$.
\end{pro}

\begin{proof}
First, we show that $N=\huaE\cdot_A$ is a Nijenhuis operator on the associative algebra $(\Gamma(A),\cdot_A)$. For any $X,Y\in\Gamma(A)$, we have
\begin{eqnarray*}
&&N(X)\cdot_A N(Y)-N\big(N(X)\cdot_A Y+X\cdot_A N(Y)-N(X\cdot_A Y)\big)\\
&=&X\cdot_A Y\cdot_A \huaE^2-\huaE\cdot_A\big(X\cdot_A Y\cdot_A \huaE +X\cdot_A Y\cdot_A \huaE-X\cdot_A Y\cdot_A \huaE \big)\\
&=&X\cdot_A Y\cdot_A \huaE^2-X\cdot_A Y\cdot_A \huaE^2\\
&=&0.
\end{eqnarray*}
Thus $N=\huaE\cdot_A$ is a Nijenhuis operator on the associative algebra $(\Gamma(A),\cdot_A)$.	

Then we show that $N=\huaE\cdot_A$ is a Nijenhuis operator on the Lie algebroid $(A,[-,-]_A,a_A)$. It is obvious that $N$ is a bundle map.
Since $\huaE$ is a pseudo-eventual identity on the $F$-algebroid $A$, taking $Y=\huaE$ in \eqref{eq:Eventual1}, we have
\begin{equation}\label{eq:Nij-relation4}
[X\cdot_A \huaE,\huaE]_A-[X,\huaE]_A\cdot_A \huaE=[\huaE,e]_A\cdot_A X\cdot_A \huaE.	
\end{equation}
 For any $X,Y\in\Gamma(A)$, expanding $[\huaE\cdot_A X,\huaE\cdot_A Y]_A$ using the Hertling-Manin relation and by \eqref{eq:Nij-relation4}, we have
\begin{eqnarray*}
&&[N(X),N(Y)]_A-N\big([N(X), Y]_A+[X,N(Y)]_A-N([X, Y]_A)\big)\\
&=&[\huaE\cdot_A X,\huaE\cdot_A Y]_A-\huaE\cdot_A\big([\huaE\cdot_A X, Y]_A+[X,\huaE\cdot_A Y]_A-\huaE\cdot_A [X, Y]_A\big)\\
&=&[\huaE\cdot_A X,\huaE]_A\cdot_A Y-[X,\huaE]_A\cdot_A \huaE\cdot_A Y-[\huaE\cdot_A Y,\huaE]_A\cdot_A X+[Y,\huaE]_A\cdot_A \huaE\cdot_A X\\
&=&([\huaE\cdot_A X,\huaE]_A-[X,\huaE]_A\cdot_A \huaE)\cdot_A Y-([\huaE\cdot_A Y,\huaE]_A-[Y,\huaE]_A\cdot_A \huaE)\cdot_A X\\
&=&[\huaE,e]_A\cdot_A X\cdot_A \huaE \cdot_A Y-[\huaE,e]_A\cdot_A Y\cdot_A \huaE \cdot_A X\\
&=&0.
\end{eqnarray*}
Thus $N=\huaE\cdot_A$ is a Nijenhuis operator on the Lie algebroid $(A,[-,-]_A,a_A)$. Therefore, $N=\huaE\cdot_A$ is a Nijenhuis operator on the $F$-algebroid $A$.

The second claim follows from Theorem \ref{pro:Nijenhuis operator}.
\end{proof}

\begin{cor}
Let $(M,\bullet)$ be an $F$-manifold with an identity $e$ and $\huaE$ a pseudo-eventual identity on $M$. Then there is a new $F$-algebroid  structure on $TM$ given by
\begin{eqnarray*}
X\bullet_\huaE Y&=&X\bullet Y\bullet \huaE,\\	
{[X,Y]}_\huaE &=&[\huaE\bullet X,Y]_{\frkX(M)}+[X,\huaE\bullet Y]_{\frkX(M)}-\huaE\bullet [X,Y]_{\frkX(M)},\\
a_{\huaE}(X)&=&\huaE\bullet X,\quad\forall~X,Y\in\frkX(M).
\end{eqnarray*}

\end{cor}


\section{Pre-$F$-algebroids and eventual identities}\label{sec:pre-F-algebroid}

In this section, first we introduce the notion of a pre-$F$-algebroid, and show that a pre-$F$-algebroid gives rise to an $F$-algebroid. Then we study eventual identities on a pre-$F$-algebroid, which give new pre-$F$-algebroids. Finally, we introduce the notion of a Nijenhuis operator on a pre-$F$-algebroid, and show that a Nijenhuis operator gives rise to a deformed pre-$F$-algebroid. Finally we give some applications to integrable systems.

\subsection{Some Properties of \preFs}
\begin{defi}

Let $(\g,\cdot)$ is a commutative associative algebra and $(\g,\ast)$ is a pre-Lie algebra. Define $\Psi:\otimes^3\g\longrightarrow \g$   by
\begin{equation}\label{eq:Com-Prelie relation2}
  \Psi(x,y,z):=x\ast(y\cdot z)-(x\ast y)\cdot z-y\cdot(x\ast z).
\end{equation}
\begin{itemize}
  \item[{\rm(i)}] The triple $(\g,\ast,\cdot)$ is called a {\bf pre-$F$-manifold algebra} if
  \begin{equation}\label{eq:pseudo-pre-HM1}
 \Psi(x,y,z)= \Psi(y,x,z),\quad \forall~x,y,z\in \g,
\end{equation}
   \item[{\rm(ii)}] The triple $(\g,\ast,\cdot)$ is called a {\bf pre-Lie commutative algebra (or PreLie-Com algebra)} if
   \begin{equation}\label{eq:prelie-com relation}
  \Psi(x,y,z)=0,\quad\forall~x,y,z\in \g.
\end{equation}
\end{itemize}

\end{defi}

It is obvious that a PreLie-Com algebra is a pre-$F$-manifold algebra.

\begin{ex}{\rm (\cite{LSB})}\label{ex:pre-Lie com algebra1}
 Let $(\g,\cdot)$ be a commutative associative algebra with a derivation $D$. Then the new product
  \begin{eqnarray*}
    x\ast y&=&x\cdot D (y),\quad\forall~x,y\in \g
  \end{eqnarray*}
 makes $(\g,\ast,\cdot)$ being a PreLie-Com algebra. Furthermore, $(\g,[-,-],\cdot)$ is an $F$-manifold algebra, where the bracket is given by
  $$ [x,y]=x\ast y-y\ast x=x\cdot D( y)-y\cdot D(x),\quad\forall~x,y\in \g.$$
\end{ex}

Let $\g=\Real[u^1,x_2,\cdots,x_n]$ be the algebra of polynomials in $n$ variables. Denote by $\frkD_n=\{\sum_{i=1}^np_i\partial_{u^i}\mid p_i\in \g\}$ the space of derivations.

\begin{ex}{\rm (\cite{LSB})}\label{ex:poly}
{\rm
  Let $\g$ be the algebra of polynomials in $n$ variables. Define $\cdot:\frkD_n\times \frkD_n\longrightarrow \frkD_n $ and $\ast:\frkD_n\times \frkD_n\longrightarrow \frkD_n$ by
  \begin{eqnarray*}
  (p\partial_{u^i})\cdot (q\partial_{u^j})&=&(pq)\delta_{ij}\partial_{u^i},\\
  (p\partial_{u^i})\ast (q\partial_{u^j})&=&p\partial_{u^i}(q)\partial_{u^j},\quad\forall~p,q\in \g.
  \end{eqnarray*}
  Then $(\frkD_n,\ast,\cdot)$ is a PreLie-Com algebra  with the identity
  $e=\partial_{u^1}+\cdots \partial_{x_n}.$
  Furthermore, it follows that $(\frkD_n,[-,-],\cdot)$ is an $F$-manifold algebra  with the identity $e$, where the bracket is given by
  $$ [p\partial_{u^i},q\partial_{u^j}]=p\partial_{u^i}(q)\partial_{u^j}-q\partial_{u^j}(p)\partial_{u^i},\quad\forall~p,q\in \g.$$}
\end{ex}

\emptycomment{\begin{ex}{\rm (\cite{LSB})}\label{ex:poly2}\rm{
 Let $\g=\Real[u^1,x_2]$ be the algebra of polynomials in two variables. Besides the PreLie-Com algebra structure given in Example~\ref{ex:poly} on
 $\frkD_2$, there is another PreLie-Com algebra  $(\frkD_2,\ast,\cdot)$ with the identity $\partial_{u^1}$, where the operations $\cdot$ and  $\ast$ are determined by
\begin{eqnarray*}
 \partial_{u^1}\cdot \partial_{u^1}=\partial_{u^1},~\quad
 \partial_{u^1}\cdot \partial_{x_2}&=& \partial_{x_2}\cdot \partial_{u^1}=\partial_{x_2},~
  \quad \partial_{x_2}\cdot \partial_{x_2}=0,\\
  (p\partial_{u^i})\ast (q\partial_{u^j})&=&p\partial_{u^i}(q)\partial_{u^j},\quad\forall~p,q\in \g~(i,j=1,2).
  \end{eqnarray*}
Furthermore, $(\frkD_2,[-,-],\cdot)$ is an $F$-manifold algebra  with the identity $\partial_{u^1}$, where the bracket is given by
  $$ [p\partial_{u^i},q\partial_{u^j}]=p\partial_{u^i}(q)\partial_{u^j}-q\partial_{u^j}(p)\partial_{u^i},\quad\forall~p,q\in \g.$$}
\end{ex}}

\begin{defi}
A {\bf \preF} is a vector bundle $A$ over $M$ equipped with bilinear operations $\cdot_A:\Gamma(A)\times \Gamma(A)\rightarrow \Gamma(A)$,  $\ast_A:\Gamma(A)\times \Gamma(A)\rightarrow \Gamma(A)$, and a bundle map $a_A:A\rightarrow TM$, called the anchor, such that $(A,\ast_A,a_A)$ is a pre-Lie algebroid, $(A,\cdot_A)$ is a commutative associative algebroid and $(\Gamma(A),[-,-]_A,\cdot_A)$ is a pre-$F$-manifold algebra. In particular, if $(\Gamma(A),\ast_A,\cdot_A)$ is a PreLie-Com algebra, we call this \preF~ a {\bf PreLie-Com algebroid}.

\end{defi}
We denote a \preF~(or PreLie-Com algebroid) by $(A,\ast_A,\cdot_A,a_A)$.

\begin{defi}
Let $(A,\ast_A,\cdot_A,a_A)$ and $(B,\ast_B,\cdot_B,a_B)$ be \preFs~ on $M$.  A bundle map $\varphi:A\longrightarrow B$ is
called a {\bf homomorphism}  of \preFs, if the following
conditions are satisfied:
\begin{eqnarray*}
  \varphi(X\cdot_A Y)=\varphi(X)\cdot_B \varphi(Y),\quad
   \varphi(X\ast_A Y)=\varphi(X)\ast_B \varphi(Y),\quad
   a_B\circ\varphi =a_A,\quad\forall~X,Y\in\Gamma(A).
\end{eqnarray*}
\end{defi}


 \begin{pro}\label{pro:preLie-Com construction}
Let $(A,\ast_A,\cdot_A,a_A)$ be a \preF. Then $(A,[-,-]_A,\cdot_A, a_A)$ is an $F$-algebroid, and denoted by
$A^c$, called the {\bf sub-adjacent $F$-algebroid} of
the \preF, where the bracket $[-,-]_A$ is given by
\begin{equation}\label{eq:Lie-bracket}
[X,Y]_A=X\ast_A Y-Y\ast_A X,\quad\forall~X,Y\in \Gamma(A).
\end{equation}
\end{pro}
\begin{proof}
  Since $(A,\ast_A,a_A)$ is a pre-Lie algebroid, $(A,[-,-]_A,a_A)$ is a Lie algebroid (\cite{LSBC}). Since $(\Gamma(A),\ast_A,\cdot_A)$ is a pre-$F$-manifold algebra,  $(\Gamma(A),[-,-]_A,\cdot_A)$ is an $F$-manifold algebra (\cite{Dot}).  Thus $(A,[-,-]_A,\cdot_A, a_A)$ is an $F$-algebroid.
\end{proof}

The notion of an $F$-manifold with a compatible flat connection was introduced by Manin in \cite{Manin1}. Recall that an {\bf   $F$-manifold with a compatible flat connection (PreLie-Com manifold)} is a triple $(M,\nabla,\bullet)$, where $M$ is a manifold, $\nabla$ is a flat connection and $\bullet$ is a $C^\infty(M)$-bilinear,  commutative, associative multiplication on the tangent bundle $TM$  such that $(TM,\nabla,\bullet,\Id)$ is a \preF~(PreLie-Com algebroid).    It is obvious that  an $F$-manifold with a compatible flat connection is a special case of \preFs.
An $F$-manifold with a compatible flat connection (resp. PreLie-Com manifold) is called {\bf semi-simple} if its sub-adjacent $F$-manifold is semi-simple.

\begin{pro}
  Let $(M,\nabla,\bullet,e)$ be a semi-simple PreLie-Com manifold with the canonical local coordinate systems $(u^1,\cdots,u^n)$. Then we have
  $$
   { \nabla_{\frac{\partial}{\partial u^i}}{\frac{\partial}{\partial u^j}}}=0,\quad i,j\in\{1,2,\cdots,n\}.
$$
\end{pro}
\begin{proof}
 Set $ \nabla_{\frac{\partial}{\partial u^i}}{\frac{\partial}{\partial u^j}}=\sum_{k}\Gamma_{ij}^k\frac{\partial}{\partial x_k}$. By \eqref{eq:prelie-com relation}, for any $i,j,k\in \{1,2,\cdots,n\}$, we have
\begin{eqnarray}
\nonumber0&=&\nabla_{\frac{\partial}{\partial u^i}}(\frac{\partial}{\partial u^j}\bullet \frac{\partial}{\partial u^k})-(\nabla_{\frac{\partial}{\partial u^i}}\frac{\partial}{\partial u^j})\bullet \frac{\partial}{\partial u^k}-\frac{\partial}{\partial u^j}\bullet (\nabla_{\frac{\partial}{\partial u^i}}\frac{\partial}{\partial u^k})\\
\label{eq:prelie-com eq1}&=&\sum_l\delta_{jk}\Gamma^l_{ik}\frac{\partial}{\partial x_l}-\Gamma^k_{ij}\frac{\partial}{\partial u^k}-\Gamma^j_{ik}\frac{\partial}{\partial u^j}.
\end{eqnarray}
For $j\neq k$ in \eqref{eq:prelie-com eq1}, we have $\Gamma_{ij}^k=0~(j\neq k)$.  For $j=k$ in \eqref{eq:prelie-com eq1}, we have $\Gamma_{ij}^j=0$. Thus for any $i,j,k\in \{1,2,\cdots,n\}$, we have
 $\Gamma_{ij}^k=0$.
\end{proof}

We give some useful formulas that will be frequently used in the sequel.

\begin{lem}
Let $(A,\ast_A,\cdot_A,a_A)$ be a \preF. Then $\Psi(X,Y,Z)$ defined by \eqref{eq:Com-Prelie relation2} is a tensor field  of type $(3,1)$ and symmetric in all arguments. Furthermore, $ \Psi$ satisfies
  \begin{eqnarray}
   \label{eq:HM pre-Lie function2} \Psi(X\cdot_A Y,Z,W)-\Psi(X,Z,W)\cdot_A Y&=&\Psi(X\cdot_A Z,Y,W)- \Psi(X,Y,W)\cdot_A Z,\\
   \label{eq:HM pre-Lie function3}\Psi(X\cdot_A Y,Z,W)-\Psi(X\cdot_A Z,Y,W)&=&\Psi(W\cdot_A Y, X,Z)-\Psi(W\cdot_A Z, X,Y)
  \end{eqnarray}
 for all $X,Y,Z,W\in \Gamma(A)$.
\end{lem}
\begin{proof}
It is straightforward to check that  $\Psi(X,Y,Z)$ is a tensor field  of type $(3,1)$.  The symmetry of $\Psi(X,Y,Z)$ in the first two arguments is the consequence of \eqref{eq:pseudo-pre-HM1} and in the last two arguments is the consequence of the commutativity of $\cdot_A$.

By the symmetry of $\Psi$, we have
  \begin{eqnarray*}
    \Psi(X\cdot_A Y,Z,W)&=&\Psi(Z,X\cdot_A Y,W)\\
    &=&Z\ast_A ((X\cdot_A Y)\cdot_A W)-(Z\ast_A (X\cdot_A Y))\cdot_A W-(X\cdot_A Y)\cdot_A (Z\ast_A W)\\
    &=&\Psi(Z,X\cdot_A W,Y)+(Z\ast_A(X\cdot_A W))\cdot_A Y+(X\cdot_A W)\cdot_A(Z\ast Y)\\
    &&-\Psi(Z,X,Y)\cdot_A W-(Z\ast_A X)\cdot_A Y\cdot W-X\cdot_A(Z\ast_A Y)\cdot_A W\\
    &&-(X\cdot_A Y)\cdot_A (Z\ast_A W)\\
    &=&\big(Z\ast_A (X\cdot_A W)-(Z\ast_A X)\cdot_A W-X\cdot_A (Z\ast_A W)\big)\cdot_A Y\\
    &&+\Psi(Z,X\cdot_A W,Y)-\Psi(Z,X,Y)\cdot_A W\\
    &=&\Psi(X,Z,W)\cdot_A Y+\Psi(X\cdot_A W,Y,Z)- \Psi(X,Y,Z)\cdot_A W.
  \end{eqnarray*}
  Thus we have
  \begin{equation}\label{eq:F-compatible pre-Lie relation1}
     \Psi(X\cdot_A Y,Z,W)-\Psi(X,Z,W)\cdot_A Y=\Psi(X\cdot_A W,Y,Z)- \Psi(X,Y,Z)\cdot_A W.
  \end{equation}
  Interchanging $Z$ and $W$ in \eqref{eq:F-compatible pre-Lie relation1}, we have
  $$ \Psi(X\cdot_A Y,W,Z)-\Psi(X,W,Z)\cdot_A Y=\Psi(X\cdot_A Z,Y,W)- \Psi(X,Y,W)\cdot_A Z.$$
 By the symmetry of $\Psi$, \eqref{eq:HM pre-Lie function2} follows.

 By \eqref{eq:HM pre-Lie function2}, we have
 \begin{eqnarray*}
  \Psi(X\cdot_A Y,Z,W)-\Psi(X\cdot_A Z,Y,W)&=&\Psi(X,Z,W)\cdot_A Y-\Psi(X,Y,W)\cdot_A Z,\\
  \Psi(W\cdot_A Y,X,Z)-\Psi(W\cdot_A Z,X,Y)&=&\Psi(W,X,Z)\cdot_A Y-\Psi(W,X,Y)\cdot_A Z.
 \end{eqnarray*}
By the symmetry of $\Psi$,  we have
$$\Psi(X,Z,W)\cdot_A Y-\Psi(X,Y,W)\cdot_A Z=\Psi(W,X,Z)\cdot_A Y-\Psi(W,X,Y)\cdot_A Z.$$
Thus \eqref{eq:HM pre-Lie function3} holds.
\end{proof}

\begin{lem}
Let $(A,\ast_A,\cdot_A,a_A)$ be a \preF~with an identity $e$. Then we have
\begin{eqnarray}
\label{eq:identity property2}\Psi(e,X,Y)&=&-(X\ast_A e)\cdot_A Y,\\
\label{eq:identity property1}(X\ast_A e)\cdot_A Y&=&(Y\ast_A e)\cdot_A X,\quad\forall~X,Y\in \Gamma(A).
\end{eqnarray}
\end{lem}
\begin{proof}
\eqref{eq:identity property2} follows by a direct calculation. By the symmetry of $\Psi$ and \eqref{eq:identity property2}, \eqref{eq:identity property1} follows.
\end{proof}

\begin{lem}\label{lem:preLie-com}
	 Let $(A,\ast_A,\cdot_A,a_A)$ be a  PreLie-Com algebroid with an identity $e$. Then we have
	\begin{equation}\label{eq:main equation1}
		X\ast_A e=0,\quad\forall~X\in\Gamma(A).	
	\end{equation}
\end{lem}
\begin{proof}
	The conclusion follows from the following relation
 $$X\ast_A (e\cdot_A e)-(X\ast_A e)\cdot_A e-(X\ast_A e)\cdot_A e=0.\qedhere$$
\end{proof}

\begin{ex}
Let $\{u\}$ be a coordinate system of $\Real$. Define an anchor map $a:T\Real\longrightarrow T\Real$, a multiplication $\cdot:\frkX(\Real)\times\frkX(\Real)\longrightarrow\frkX(\Real)$ and a multiplication  $\ast:\frkX(\Real)\times \frkX(\Real)\longrightarrow \frkX(\Real)$   by
\begin{eqnarray*}
a(f\frac{\partial}{\partial u})=uf\frac{\partial}{\partial u},\quad
{f\frac{\partial}{\partial u}\cdot g\frac{\partial}{\partial u}}=fg\frac{\partial}{\partial u}, \quad
{f\frac{\partial}{\partial u}\ast g\frac{\partial}{\partial u}}=uf\frac{\partial g}{\partial u} \frac{\partial}{\partial u}, 
\end{eqnarray*}	
for all $f,g\in C^{\infty}(\Real)$. Then $(T\Real,\ast,\cdot,a)$ is a PreLie-Com algebroid with the identity $\frac{\partial}{\partial u}$. Furthermore, $(T\Real,[-,-],\cdot,a)$ is an $F$-algebroid with the identity $\frac{\partial}{\partial u}$, where $[-,-]$ is given by
$$[{f\frac{\partial}{\partial u},g\frac{\partial}{\partial u}}]=u(f\frac{\partial g}{\partial u}-g\frac{\partial f}{\partial u})\frac{\partial }{\partial u}.$$
\end{ex}

\begin{defi}
Let $(\frkg,\ast,\cdot)$ be a pre-$F$-manifold algebra (PreLie-Com algebra). An {\bf action}
of $\frkg$ on a manifold $M$ is a linear map $\rho:\frkg\longrightarrow\frkX(M)$  from $\g$ to the space of vector fields on $M$,
such that for all $x,y\in\frkg$, we have
 $$
 \rho(x\ast y-y\ast x)=[\rho(x),\rho(y)]_{\frkX(M)}.
 $$
\end{defi}

 Given an action of a pre-$F$-manifold algebra (PreLie-Com algebra) $\frkg$ on $M$, let $A=M\times
\g$ be the trivial bundle. Define an anchor map $a_\rho:A\longrightarrow TM$, a multiplication $\cdot_\rho:\Gamma(A)\times \Gamma(A)\longrightarrow \Gamma(A)$ and a bracket $\ast_\rho:\Gamma(A)\times \Gamma(A)\longrightarrow \Gamma(A)$   by
\begin{eqnarray}
a_\rho(m,u)&=&\rho(u)_m,\quad \forall ~m\in M, u\in\g,\label{action P1}\\
{X\cdot_\rho Y}&=&X\cdot Y, \label{action P2}\\
{X\ast_\rho Y}&=&\huaL_{\rho(X)}Y+X\ast Y,\quad \forall~X,Y\in\Gamma(A), \label{action P3}
\end{eqnarray}
where  $X\cdot Y$ and $X\ast Y$ are the pointwise $C^{\infty}(M)$-bilinear multiplication and bracket, respectively.

\begin{pro}
With the above notations, $(A=M\times\frkg,\ast_\rho,\cdot_\rho,a_\rho)$ is a \preF~(PreLie-Com algebroid), which we call  an {\bf action \preF}~(action PreLie-Com algebroid),
where  $\ast_\rho$, $\cdot_\rho $ and $a_\rho$ are given by $(\ref{action P3})$, $(\ref{action P2})$
and $(\ref{action P1})$, respectively. 
\end{pro}
\begin{proof}
It follows by a similar proof of Proposition \ref{pro:action F-algebroid}.
\end{proof}

It is obvious that the sub-adjacent $F$-algebroid of the action \preF~is an action $F$-algebroid.

\begin{ex}{\rm
Consider the PreLie-Com algebra $(\frkD_n,\cdot,\ast)$ given by Example \ref{ex:poly}. Let $(t_1,\cdots,t_n)$ be the canonical coordinate systems on $\Real^n$. Define a map $\rho:\frkD_n\longrightarrow \frkX(\Real^n)$ by
$$\rho(p(u^1,\cdots,u^n)\partial_{u^i})=p(t_1,\cdots,t_n)\frac{\partial}{\partial t_i},\quad i\in\{1,2,\cdots,n\}.$$
It is straightforward to check that $\rho$ is an action of the PreLie-Com algebra $\frkD_n$ on $\Real^n$. Thus $(A=\Real^n\times\frkD_n,\ast_\rho,\cdot_\rho,a_\rho)$ is a PreLie-Com algebroid,
where  $\ast_\rho$, $\cdot_\rho $ and $a_\rho$ are given by
\begin{eqnarray*}
a_\rho(m,p(u^1,u^2,\cdots,u^n)\partial_{u^i})&=&p(m)\frac{\partial}{\partial t_i}\mid_m,\quad \forall ~m\in \Real^n,\\
{(f\otimes (p\partial_{u^i}))\cdot_\rho (g\otimes (q\partial_{u^j}))}&=&(fg)\otimes (pq\delta_{ij}\partial_{u^i} ), \\
{(f\otimes (p\partial_{u^i}))\ast_\rho (g\otimes (q\partial_{u^j}))}&=&fp\frac{\partial g}{\partial t_i}\otimes (q\partial_{u^j})+(fg)\otimes p\partial_{u^i}(q)\partial_{u^j},
\end{eqnarray*}
where $f,g\in C^{\infty}(\Real^n),p,q\in\Real[u^1,\cdots,u^n]$.}
 \end{ex}

\emptycomment{\begin{ex}{\rm
Consider the PreLie-Com algebra $(\frkD_2,\cdot,\ast)$ given by Example \ref{ex:poly2}. Let $(t_1,t_2)$ be the canonical coordinate systems on $\Real^2$. There is an action $\rho:\frkD_2\longrightarrow \frkX(\Real^2)$ of the PreLie-Com algebra $\frkD_2$ on $\Real^2$ given by
$$\rho(p(u^1,x_2)\partial_{u^i})=p(t_1,t_2)\frac{\partial}{\partial t_i},\quad i\in\{1,2\}.$$
Then $(A=\Real^2\times\frkD_2,\ast_\rho,\cdot_\rho,a_\rho)$ is a PreLie-Com algebroid,
where  $\ast_\rho$, $\cdot_\rho $ and $a_\rho$ are given by
\begin{eqnarray*}
a_\rho(m,p(u^1,x_2)\partial_{u^i})&=&p(m)\frac{\partial}{\partial t_i}\mid_m,\quad \forall ~m\in \Real^2,\\
{(f\otimes p\partial_{u^1})\cdot_\rho (g\otimes q\partial_{u^i})}&=&(fg)\otimes (pq\partial_{u^i} ), \\
{(f\otimes p\partial_{x_2})\cdot_\rho (g\otimes q\partial_{x_2})}&=&0,\\
{(f\otimes p\partial_{u^i})\ast_\rho (g\otimes q\partial_{u^j})}&=&fp\frac{\partial g}{\partial t_i}\otimes q\partial_{u^j}+(fg)\otimes p\partial_{u^i}(q)\partial_{u^j},
\end{eqnarray*}
where $f,g\in C^{\infty}(\Real^2),p,q\in\Real[u^1,x_2],~i,j\in\{1,2\}$.}
 \end{ex}}

\subsection{Eventual identities of \preFs}
\begin{defi}
 Let $(A,\ast_A,\cdot_A,a_A)$ be a \preF~with an identity $e$. A section $\huaE\in\Gamma(A)$  is called a {\bf pseudo-eventual identity} on $A$ if the following equalities hold:
 \begin{eqnarray}
  \label{eq:Eventua31} \Psi(\huaE, X,Y)&=& -(\huaE\ast_A e)\cdot_A X\cdot_A Y,\\
\label{eq:Eventua32} (X\ast_A \huaE)\cdot_A Y&=& (Y\ast_A \huaE)\cdot_A X,\quad\forall~X,Y\in \Gamma(A).
 \end{eqnarray}

 A pseudo-eventual identity $\huaE$ on the \preF~with an identity $e$ is called an {\bf eventual identity} if it is invertible.
\end{defi}




\begin{pro}\label{pro:eventual identity F-pre}
  Let $(A,\ast_A,\cdot_A,e,a_A)$ be a \preF~with an identity $e$. If $\huaE\in\Gamma(A)$ is a pseudo-eventual identity on $A$, then $\huaE\in\Gamma(A)$ is a pseudo-eventual identity on its sub-adjacent $F$-algebroid $A^c$.
\end{pro}
\begin{proof}
  By a direct calculation, for $X,Y\in\Gamma(A)$, we have
  \begin{eqnarray*}
  &&P_\huaE(X,Y)-[e,\huaE]_A\cdot_A X\cdot_A Y\\
  &=&\huaE\ast_A (X\cdot_A Y)-(X\cdot_A Y)\ast_A \huaE-(\huaE\ast_A X)\cdot_A Y+( X\ast_A\huaE)\cdot_A Y\\
  &&-(\huaE\ast_A Y)\cdot_A X+( Y\ast_A\huaE)\cdot_A X-(e\ast_A \huaE)\cdot_A X\cdot_A Y+(\huaE\ast_A e)\cdot_A X\cdot_A Y\\
  &=&\Psi(\huaE, X,Y)+(\huaE\ast_A e)\cdot_A X\cdot_A Y-(X\cdot_A Y)\ast_A \huaE+( X\ast_A\huaE)\cdot_A Y\\
  &&+( Y\ast_A\huaE)\cdot_A X-(e\ast_A \huaE)\cdot_A X\cdot_A Y.
  \end{eqnarray*}
  By \eqref{eq:Eventua31} and \eqref{eq:Eventua32}, we have
 \begin{eqnarray*}\label{eq:eventual property}
&&P_\huaE(X,Y)-[e,\huaE]_A\cdot_A X\cdot_A Y\\
&=&-(X\cdot_A Y)\ast_A \huaE+( X\ast_A\huaE)\cdot_A Y+( Y\ast_A\huaE)\cdot_A X-(e\ast_A \huaE)\cdot_A X\cdot_A Y\\
 &=&-(e\ast_A \huaE)\cdot_A X\cdot_A Y+( X\ast_A\huaE)\cdot_A Y+( X\ast_A\huaE)\cdot_A Y-(e\ast_A \huaE)\cdot_A X\cdot_A Y\\
 &=&2( X\ast_A\huaE)\cdot_A Y-2(e\ast_A \huaE)\cdot_A X\cdot_A Y\\
 &=&2( X\ast_A\huaE)\cdot_A Y-2( X\ast_A\huaE)\cdot_A Y
 \\&=&0.    \end{eqnarray*}
  Thus $\huaE\in\Gamma(A)$ is a pseudo-eventual identity on its sub-adjacent $F$-algebroid $A^c$.
\end{proof}

By Lemma \ref{lem:preLie-com}, we have
\begin{pro}\label{pro:pre-Com identities}
   Let $(A,\ast_A,\cdot_A,a_A)$ be a \preF~with an identity $e$ and $\huaE$ an invertible element in $\Gamma(A)$. If $(A,\ast_A,\cdot_A,a_A)$ is a PreLie-Com algebroid, then $\huaE$ is an eventual identity on $A$ if and only if \eqref{eq:Eventua32} holds.
\end{pro}

\begin{lem}\label{lem:important}
 Let $(A,\ast_A,\cdot_A,e,a_A)$ be a \preF. Then for $\huaE\in\Gamma(A)$, \eqref{eq:Eventua31} holds if and only if 	
 \begin{equation}\label{eq:pre-Lie eventual1}
 \Psi(X,\huaE\cdot_A Y,Z)=\Psi(Y,\huaE\cdot_A X,Z),\quad\forall~X,Y,Z\in \Gamma(A).\end{equation}
\end{lem}
\begin{proof}
Assume that \eqref{eq:pre-Lie eventual1} holds. By \eqref{eq:HM pre-Lie function2}, we have
\begin{eqnarray}\label{eq:pre-Lie eventual2}
 \Psi(\huaE,X,Z)\cdot_A Y- \Psi(\huaE,Y,Z)\cdot_A X=\Psi(X,\huaE\cdot_A Y,Z)-\Psi(Y,\huaE\cdot_A X,Z)=0.
\end{eqnarray}
Taking $Y=e$ in \eqref{eq:pre-Lie eventual2}, we have
$$\Psi(\huaE,X,Z)=-(\huaE\ast_A e)\cdot_A X\cdot_A Z.$$
This implies that \eqref{eq:Eventua31} holds.	

Conversely, if \eqref{eq:Eventua31} holds, then we have
\begin{eqnarray*}
 \Psi(\huaE,X,Z)\cdot_A Y- \Psi(\huaE,Y,Z)\cdot_A X=-(\huaE\ast_A e)\cdot_A X\cdot_A Z\cdot_A Y+(\huaE\ast_A e)\cdot_A Y\cdot_A Z\cdot_A X=0.\end{eqnarray*}
By \eqref{eq:HM pre-Lie function2}, we have
$$\Psi(X,\huaE\cdot_A Y,Z)=\Psi(Y,\huaE\cdot_A X,Z).$$
This implies that \eqref{eq:pre-Lie eventual1} holds.
	\end{proof}

Denote the set of all pseudo-eventual identities on a \preF~ $(A,\ast_A,\cdot_A,a_A)$ with an identity $e$ by $\frkE(A)$.
\begin{pro}\label{lem:Eventual1}
Let $(A,\ast_A,\cdot_A,a_A)$ be a \preF~ with an identity $e$. Then for any   $\huaE_1,\huaE_2\in\frkE(A)$, $\huaE_1\cdot_A \huaE_2\in\frkE(A).$ Furthermore, if $\huaE$ is an eventual identity on $A$, then $\huaE^{-1}$ is also an eventual identity on $A$.
\end{pro}
\begin{proof}
For the first claim, let $\huaE_1,\huaE_2$ be two pseudo-eventual identities on the \preF~ $A$. For all $X,Y,Z\in\Gamma(A)$, since $\huaE_1$ and $\huaE_2$ are pseudo-eventual identities, by \eqref{eq:pre-Lie eventual1}, we have
\begin{eqnarray*}
	\Psi(X,\huaE_1\cdot_A \huaE_2\cdot_A Y,Z)&=&\Psi(\huaE_2\cdot_A Y,\huaE_1\cdot_A X ,Z);\\
	\Psi(Y,\huaE_2\cdot_A \huaE_1\cdot_A X,Z)&=&\Psi(\huaE_1\cdot_A X ,\huaE_2\cdot_A Y,Z).	\end{eqnarray*}
Thus by the symmetry of $\Psi$, we have
\begin{eqnarray*}
	\Psi(X,\huaE_1\cdot_A \huaE_2\cdot_A Y,Z)=\Psi(Y,\huaE_1\cdot_A \huaE_2\cdot_A X,Z).
\end{eqnarray*}
By Lemma \ref{lem:important}, we have
$$\Psi(\huaE_1\cdot_A \huaE_2, X,Y)=-((\huaE_1\cdot_A \huaE_2)\ast_A e)\cdot_A X\cdot_A Y.$$
For all $X,Y\in\Gamma(A)$, by \eqref{eq:pseudo-pre-HM1}, we have
\begin{eqnarray*}
	&&(X\ast_A(\huaE_1\cdot_A \huaE_2))\cdot_A Y-	(Y\ast_A(\huaE_1\cdot_A \huaE_2))\cdot_A X\\
	&=&\Psi(\huaE_1,X,\huaE_2)\cdot_A Y+(X\ast_A \huaE_1)\cdot_A \huaE_2\cdot_A Y+(X\ast_A \huaE_2)\cdot_A \huaE_1\cdot_A Y\\
	&&-\Psi(\huaE_1,Y,\huaE_2)\cdot_A X-(Y\ast_A \huaE_1)\cdot_A \huaE_2\cdot_A X-(Y\ast_A \huaE_2)\cdot_A \huaE_1\cdot_A X.
	\end{eqnarray*}
By \eqref{eq:HM pre-Lie function2} and \eqref{eq:pre-Lie eventual1}, we have
\begin{eqnarray*}
\Psi(\huaE_1,X,\huaE_2)\cdot_A Y-\Psi(\huaE_1,Y,\huaE_2)\cdot_A X=	\Psi(\huaE_1\cdot_A Y,X,\huaE_2)-\Psi(\huaE_1\cdot_A X,Y,\huaE_2)=0.
\end{eqnarray*}
Using the above relation and by \eqref{eq:Eventua32}, we have
\begin{eqnarray*}
	&&(X\ast_A(\huaE_1\cdot_A \huaE_2))\cdot_A Y-	(Y\ast_A(\huaE_1\cdot_A \huaE_2))\cdot_A X\\
	&=&(X\ast_A \huaE_1)\cdot_A \huaE_2\cdot_A Y+(X\ast_A \huaE_2)\cdot_A \huaE_1\cdot_A Y-(Y\ast_A \huaE_1)\cdot_A \huaE_2\cdot_A X-(Y\ast_A \huaE_2)\cdot_A \huaE_1\cdot_A X\\
	&=&(Y\ast_A \huaE_1)\cdot_A \huaE_2\cdot_A X+(Y\ast_A \huaE_2)\cdot_A \huaE_1\cdot_A X-(Y\ast_A \huaE_1)\cdot_A \huaE_2\cdot_A X-(Y\ast_A \huaE_2)\cdot_A \huaE_1\cdot_A X\\
&=&0,
\end{eqnarray*}
which implies that
$$(X\ast_A(\huaE_1\cdot_A \huaE_2))\cdot_A Y=(Y\ast_A(\huaE_1\cdot_A \huaE_2))\cdot_A X.$$
Thus $\huaE_1\cdot_A \huaE_2\in\frkE(A).$

For the second claim, using relation \eqref{eq:pre-Lie eventual1} with $X$ and $Y$ replaced by $\huaE^{-1}\cdot_A X$ and $\huaE^{-1}\cdot_A Y$ respectively, we have
\begin{eqnarray*}
 0&=& \Psi(\huaE^{-1}\cdot_A X,\huaE\cdot_A \huaE^{-1}\cdot_A Y,Z)- \Psi(\huaE^{-1}\cdot_A Y,\huaE\cdot_A \huaE^{-1}\cdot_A X,Z)\\
 &=&\Psi(\huaE^{-1}\cdot_A X,Y,Z)- \Psi(\huaE^{-1}\cdot_A Y,X,Z).
\end{eqnarray*}
By the symmetry of $\Psi$ and Lemma \ref{lem:important}, we have
$$\Psi(\huaE^{-1}, X,Y)= -(\huaE^{-1}\ast_A e)\cdot_A X\cdot_A Y.$$

By \eqref{eq:HM pre-Lie function2} and \eqref{eq:pre-Lie eventual1}, we have
\begin{equation}\label{eq:Eventual property2}
\Psi(X,\huaE,\huaE^{-1})\cdot_A Y=\Psi(Y,\huaE,\huaE^{-1})\cdot_A X.
\end{equation}
Furthermore, by a direct calculation, we have
\begin{eqnarray*}
  (X\ast_A \huaE^{-1})\cdot_A Y\cdot_A \huaE&=&\Psi(X,\huaE,\huaE^{-1})\cdot_A Y-(X\ast_A e)\cdot_A Y+(X\ast_A \huaE)\cdot_A Y\cdot_A \huaE^{-1},\\
   (Y\ast_A \huaE^{-1})\cdot_A X\cdot_A \huaE&=&\Psi(Y,\huaE,\huaE^{-1})\cdot_A X-(Y\ast_A e)\cdot_A X+(Y\ast_A \huaE)\cdot_A X\cdot_A \huaE^{-1}.
\end{eqnarray*}
By \eqref{eq:identity property1}, \eqref{eq:Eventua32} and \eqref{eq:Eventual property2}, we have
$$ (X\ast_A \huaE^{-1})\cdot_A Y\cdot_A \huaE= (Y\ast_A \huaE^{-1})\cdot_A X\cdot_A \huaE.$$
Because $\huaE$ is invertible, we have
$$(X\ast_A \huaE^{-1})\cdot_A Y= (Y\ast_A \huaE^{-1})\cdot_A X.$$
Thus $\huaE^{-1}$ is an eventual identity on $A$.
\end{proof}

\begin{pro}\label{thm:construction F-algebroid3}
   Let $(A,\ast_A,\cdot_A,a_A)$ be a \preF~with an identity $e$. Then $\huaE$ is a pseudo-eventual identity on $A$ if and only if $(A,\ast_A,\cdot_{\huaE},a_A)$ is a \preF, where $\cdot_{\huaE}:\Gamma(A)\times \Gamma(A)\longrightarrow \Gamma(A)$ is given by \eqref{eq:New associative mult}.
\end{pro}
\begin{proof}
  Define
  $$\tilde{\Psi}(X,Y,Z)=X\ast_A(Y\cdot_\huaE Z)-(X\ast_A Y)\cdot_\huaE Z-Y\cdot_\huaE(X\ast_A Z),\quad\forall~X,Y,Z\in \Gamma(A).$$
  By a straightforward computation, we have
  \begin{eqnarray}
    \label{eq:UHM-pre 1}\tilde{\Psi}(X,Y,Z)=\Psi(X,\huaE\cdot_A Y,Z)+\Psi(X,\huaE,Y)\cdot_A Z+(X\ast_A \huaE)\cdot_A Y\cdot_A Z,\\
    \label{eq:UHM-pre 2}\tilde{\Psi}(Y,X,Z)=\Psi(Y,\huaE\cdot_A X,Z)+\Psi(Y,\huaE,X)\cdot_A Z+(Y\ast_A \huaE)\cdot_A X\cdot_A Z.
  \end{eqnarray}
 By the symmetry of $\Psi$, $(A,\ast_A,\cdot_\huaE,a_A)$ is a \preF~ if and only if
 \begin{equation}\label{eq:HM pre-Lie relation}
 \Psi(X,\huaE\cdot_A Y,Z)-\Psi(Y,\huaE\cdot_A X,Z)=(Y\ast_A \huaE)\cdot_A X\cdot_A Z-(X\ast_A \huaE)\cdot_A Y\cdot_A Z.
 \end{equation}
By the symmetry of $\Psi$ and \eqref{eq:HM pre-Lie function3}, we have
  \begin{eqnarray*}
   \Psi(X,\huaE\cdot_A Y,e)-\Psi(Y,\huaE\cdot_A X,e)= \Psi(e\cdot_A Y,\huaE,X)-\Psi(e\cdot_A X,\huaE,Y)=0.
  \end{eqnarray*}
 Taking $Z=e$ in \eqref{eq:HM pre-Lie relation}, we have
 \begin{equation*}
 (X\ast_A \huaE)\cdot_A Y= (Y\ast_A \huaE)\cdot_A X.
 \end{equation*}
This implies that \eqref{eq:Eventua32} holds. Furthermore, by \eqref{eq:Eventua32},  \eqref{eq:HM pre-Lie relation} implies that \eqref{eq:pre-Lie eventual1} holds. By Lemma \ref{lem:important}, \eqref{eq:pre-Lie eventual1} is equivalent to \eqref{eq:Eventua31}. Thus $\huaE$ is a pseudo-eventual identity on $(A,\ast_A,\cdot_A,e,a_A)$.

Conversely, if $\huaE$ is a pseudo-eventual identity on $(A,\ast_A,\cdot_A,e,a_A)$, by Lemma \ref{lem:important}, we have
$$\Psi(X,\huaE\cdot_A Y,Z)=\Psi(Y,\huaE\cdot_A X,Z).$$
Furthermore, by \eqref{eq:Eventua32}, \eqref{eq:HM pre-Lie relation} follows. Thus $(A,\ast_A,\cdot_\huaE,a_A)$ is a \preF.
\end{proof}

\begin{cor}\label{cor:Nij-pre-F-algebroid}
Let $(M,\nabla,\bullet)$ be an $F$-manifold with a compatible flat connection and $\huaE$ a pseudo-eventual identity on $M$. Then $(M,\nabla,\bullet_\huaE)$ is also an $F$-manifold with a compatible flat connection, where $\bullet_\huaE$ is given by
\begin{eqnarray}
X\bullet_\huaE Y&=&X\bullet Y\bullet \huaE,\quad\forall~X,Y\in\frkX(M).
\end{eqnarray}

\end{cor}

\begin{thm}
  Let $(A,\ast_A,\cdot_A,a_A)$ be a \preF~with an identity $e$.  Then $\huaE$ is an eventual identity on $A$ if and only if $(A,\ast_A,\cdot_{\huaE},a_A)$ is a \preF~ with the identity $\huaE^{-1}$, which is called the Dubrovin's dual of $(A,\ast_A,\cdot_A,a_A)$, where $\cdot_{\huaE}$ is given by \eqref{eq:New associative mult}.  Moreover, $e$ is an eventual identity on the \preF~ $(A,\ast_A,\cdot_{\huaE},\huaE^{-1},a_A)$  and the map
  \begin{equation}\label{eq:involution2}
  (A,\ast_A,\cdot_A,e,a_A,\huaE)\longrightarrow(A,\ast_A,\cdot_{\huaE},\huaE^{-1},a_A,e^{\dag})
  \end{equation}
  is an involution of the set of pre-$F$-algebroids  with eventual identities, where $e^{\dag}=\huaE^{-2}$ is the inverse of $e$ with respect to the multiplication $\cdot_\huaE$.
\end{thm}
\begin{proof}
By Proposition \ref{thm:construction F-algebroid3}, the first claim follows immediately. For the second claim, assume that $\huaE$ is an eventual identity on $(A,\ast_A,\cdot_A,e,a_A)$. We need to show that $e$ is an eventual identity on the \preF~ $(A,\ast_A,\cdot_{\huaE},\huaE^{-1},a_A)$, i.e.
\begin{eqnarray}
  \label{eq:Eventua211} \tilde{\Psi}(e,X,Y)&=& -(e\ast_A \huaE^{-1})\cdot_\huaE X\cdot_\huaE Y,\\
\label{eq:Eventua222} (X\ast_A e)\cdot_\huaE Y&=& (Y\ast_A e)\cdot_\huaE X.
 \end{eqnarray}
By \eqref{eq:identity property1}, we have
\begin{eqnarray*}
  (X\ast_A e)\cdot_\huaE Y-(Y\ast_A e)\cdot_\huaE X=( (X\ast_A e)\cdot_A Y-(Y\ast_A e)\cdot_A X)\cdot_A \huaE=0,
\end{eqnarray*}
which implies that \eqref{eq:Eventua222} holds.

On the one hand, by \eqref{eq:Eventua31} and \eqref{eq:pre-Lie eventual1}, we have
\begin{eqnarray*}
	\tilde{\Psi}(e,X,Y)&=&\Psi(e,\huaE\cdot_A X,Y)+\Psi(e,\huaE,X)\cdot_A Y+(e\ast_A \huaE)\cdot_A X\cdot_A Y\\
	&=&\Psi(\huaE,X,Y)+\Psi(\huaE,e,X)\cdot_A Y+(e\ast_A \huaE)\cdot_A X\cdot_A Y\\
	&=&-(\huaE\ast_A e)\cdot_A X\cdot_A Y-(\huaE\ast_A e)\cdot_A X\cdot_A Y+(e\ast_A \huaE)\cdot_A X\cdot_A Y\\
	&=&-2(\huaE\ast_A e)\cdot_A X\cdot_A Y+(e\ast_A \huaE)\cdot_A X\cdot_A Y.
\end{eqnarray*}
On the other hand, taking $X=\huaE$ and $Y=\huaE^{-1}$ in \eqref{eq:Eventua31}, by the symmetry of $\Psi$, we have
\begin{eqnarray*}
	\Psi(e,\huaE,\huaE^{-1})=\Psi(\huaE,e,\huaE^{-1})=-(\huaE\ast_A e)\cdot_A\huaE^{-1},
\end{eqnarray*}
which implies that
\begin{eqnarray*}
e\ast_A e-(e\ast_A \huaE)\cdot_A \huaE^{-1}-(e\ast_A \huaE^{-1})\cdot_A \huaE=	-(\huaE\ast_A e)\cdot_A\huaE^{-1}.
\end{eqnarray*}
Furthermore, by \eqref{eq:identity property1}, we have
\begin{eqnarray*}
	(e\ast_A \huaE^{-1})\cdot_A \huaE^2&=&(e\ast_A e)\cdot_A \huaE-e\ast_A \huaE+\huaE\ast_A e=2\huaE\ast_A e-e\ast_A \huaE.
\end{eqnarray*}
Thus we have
$$\tilde{\Psi}(e,X,Y)=-(e\ast_A \huaE^{-1})\cdot_A \huaE^2\cdot_A X\cdot_A Y=-(e\ast_A \huaE^{-1})\cdot_\huaE X\cdot_\huaE Y.$$
Therefore, $e$ is an eventual identity on the $F$-algebroid with a compatible pre-Lie algebroid $(A,\ast_A,\cdot_{\huaE},\huaE^{-1},a_A)$.

By Proposition \ref{lem:Eventual1},  $e^{\dag}=\huaE^{-2}$ is an eventual identity on the \preF~ $(A,\ast_A,\cdot_{\huaE},\huaE^{-1},a_A)$. Then similar to the proof of Theorem \ref{thm:involution1}, the map given by \eqref{eq:involution2} is an involution of the set of \preFs~  with eventual identities.
\end{proof}

\begin{ex}
  Consider the PreLie-Com algebra $(\g,\ast,\cdot)$ with an identity $e$ given by Example \ref{ex:pre-Lie com algebra1}. By a direct calculation, for any $\huaE\in\g$, we have
  \begin{eqnarray*}
    (x\ast \huaE)\cdot y-(y\ast \huaE)\cdot x=x\cdot D(\huaE)\cdot y-y\cdot D(\huaE)\cdot x=0,\quad\forall~x,y\in\g.
  \end{eqnarray*}
 By Proposition \ref{pro:pre-Com identities}, $\huaE$ is a pseudo-eventual identity on $\g$. Thus any element of $\g$ is a pseudo-eventual identity on $\g$. Furthermore, for any $\huaE\in\g$, there is a new pre-$F$-manifold algebra structure on $\g$ given by
  \begin{eqnarray*}
    x\cdot_\huaE y=x\cdot y\cdot \huaE,\quad x\ast y=x\cdot D(y),\quad \forall~x,y\in\g.
  \end{eqnarray*}
\end{ex}

\begin{ex}{\rm
 Let $(M,\nabla,\bullet,e)$ be a semi-simple PreLie-Com manifold with local coordinate systems $(u^1,\cdots,u^n)$. Then any  pseudo-eventual identity on $TM$ is of the form
 $$\huaE=f_1(u^1)\frac{\partial}{\partial u^1}+\cdots+f_n(u^n)\frac{\partial}{\partial u^n},$$
where $f_i(u^i)\in\CWM$ depends only on $u^i$ for $i=1,2,\cdots,n$. Furthermore, if all $f_i(u^i)$ are  non-vanishing everywhere, then $\huaE\in \XM$ is an  eventual identity.}
\end{ex}

\begin{ex}
 Let $(u^1,u^2)$ be a local coordinate systems on $\Real^2$. Define two multiplications by
\begin{eqnarray*} \frac{\partial}{\partial u^1}\bullet \frac{\partial}{\partial u^i}=\frac{\partial}{\partial u^i},\quad \frac{\partial}{\partial u^2}\bullet \frac{\partial}{\partial u^2}=0,\quad
 \frac{\partial}{\partial u^i}\ast \frac{\partial}{\partial u^j}=0, \quad i,j\in\{1,2\}.\end{eqnarray*}
 Then $(T\Real^2,\ast,\bullet,\Id)$ is a PreLie-Com algebroid with the identity $\frac{\partial}{\partial u^1}$ and thus $(T\Real^2,\ast,\bullet,\Id)$ is a \preF~ with the identity $\frac{\partial}{\partial u^1}$.

Furthermore, any pseudo-eventual identity on   $(T\Real^2,\ast,\bullet,\Id)$ is of the form
$$ \huaE=f_1(u^1)\frac{\partial}{\partial u^1}+f_2(u^1,u^2)\frac{\partial}{\partial u^2}$$
with
$\frac{\partial f_1}{\partial u^1}=\frac{\partial f_2}{\partial u^2},$
where $f_1 \in C^{\infty}(\Real^2)$ depends only on $u^1$ and $f_2 $ is any smooth function. Furthermore,
any pseudo-eventual identity on the sub-adjacent $F$-algebroid of $(T\Real^2,\ast,\bullet,\Id)$ is of the form
$$ \huaE=f_1(u^1)\frac{\partial}{\partial u^1}+f_2(u^1,u^2)\frac{\partial}{\partial u^2}.$$
In particular, if $f_1(u^1)$ is  non-vanishing everywhere, then $\huaE$ is an  eventual identity on the sub-adjacent $F$-algebroid of $(T\Real^2,\ast,\bullet,\Id)$.
\end{ex}

\subsection{Nijenhuis operators and deformed \preFs}
Recall from \cite{LSBC} that a Nijenhuis operator on a pre-Lie algebroid $(A,\ast_A,a_A)$ is a bundle map $N:A\lon A$ such that
\begin{equation}\label{eq:Nij-Struc-pre}
 N(X)\ast_A N(Y)=N\big(N(X)\ast_A Y+X\ast_A N(Y)-N(X\ast_A Y)\big),\quad\forall~X,Y\in\Gamma(A).
\end{equation}

\begin{defi}
  Let $(A,\ast_A,\cdot_A,a_A)$ be a \preF. A bundle map $N:A\lon A$ is called a {\bf Nijenhuis operator} on  $(A,\ast_A,\cdot_A,a_A)$ if $N$ is both a Nijenhuis operator on the commutative associative algebra $(\Gamma(A),\cdot_A)$ and a Nijenhuis operator on the pre-Lie algebroid $(A,\ast_A,a_A)$.
\end{defi}

\begin{thm}\label{pro:pre-Lie Nijenhuis operator}
  Let $N:A\longrightarrow A$ be a Nijenhuis operator on a \preF~ $(A,\ast_A,\cdot_A,a_A)$. Then $(A,\ast_N,\cdot_N,a_N=a_A\circ N)$
 is a \preF~and $N$ is a homomorphism from the \preF~$(A,\ast_N,\cdot_N,a_N=a_A\circ N)$ to $(A,\ast_A,\cdot_A,a_A)$, where the operation $\cdot_N$ is given by \eqref{eq:deformCA} and the operation $\ast_N:\Gamma(A)\times\Gamma(A)\longrightarrow \Gamma(A)$ is given by
\begin{eqnarray}
\label{eq:deformPLA}  X\ast_N Y&=&N(X)\ast_A Y+X\ast_A N(Y)-N(X\ast_A Y),\quad\forall~X,Y\in \Gamma(A).
\end{eqnarray}
\end{thm}
\begin{proof}
Since $N$ is a Nijenhuis operator on the commutative associative algebra $(\Gamma(A),\cdot_A)$, it follows that $(\Gamma(A),\cdot_N)$ is a commutative associative algebra. Since $N$ is a Nijenhuis operator on  the pre-Lie algebroid $(A,\ast_A,a_A)$, $(A,\ast_N,a_N)$ is a pre-Lie algebroid (\cite{LSBC}).	

Define
\begin{equation}\label{eq:pre-Lie deformed equation}
\Psi_N(X,Y,Z):=X\ast_N(Y\cdot_N Z)-(X\ast_N Y)\cdot_N Z-(X\ast_N Z)\cdot_N Y,\quad\forall~X,Y,Z\in\Gamma(A).
\end{equation}
By a direct calculation, we have
\begin{eqnarray*}
  \Psi_N(X,Y,Z)&=&\Psi(NX,NY,Z)+\Psi(NX,Y,NZ)+\Psi(X,NY,NZ)\\
  &&-N\big(\Psi(NX,Y,Z)+\Psi(X,NY,Z)+\Psi(X,Y,NZ)\big)+N^2(\Psi(X,Y,Z)).
\end{eqnarray*}
Thus by \eqref{eq:pseudo-pre-HM1}, we have
$$\Psi_N(X,Y,Z)=\Psi_N(Y,X,Z).$$
This implies that $(A,\ast_N,\cdot_N,a_N=a_A\circ N)$ is a \preF. It is not hard to see that $N$ is a homomorphism from the \preF~$(A,\ast_N,\cdot_N,a_N=a_A\circ N)$ to $(A,\ast_A,\cdot_A,a_A)$.
\end{proof}

\emptycomment{\begin{pro}\label{pro:eventual identity Nij}
Let $(A,\ast_A,\cdot_A,e,a_A)$ be a \preF~with an identity $e$ and	$\huaE$   a pseudo-eventual identity on $A$. Then the endomorphism $N=\huaE\cdot_A$ is a Nijenhuis operator on the pre-Lie algebroid $(A,\ast_A,a_A)$.
\end{pro}
\begin{proof}

\end{proof}}

\begin{pro}\label{pro:eventual identity Nij}
Let $(A,\ast_A,\cdot_A,a_A)$ be a \preF~ with an identity $e$ and $\huaE$ a pseudo-eventual identity on $A$. Then the endomorphism $N=\huaE\cdot_A$ is a Nijenhuis operator on the \preF~ $(A,\ast_A,\cdot_A,a_A)$.
Furthermore, $(A,\ast_\huaE,\cdot_\huaE,a_\huaE)$ is a \preF, where the multiplication $\ast_\huaE$ is given by
\begin{equation}
X\ast_\huaE Y=(\huaE\cdot_A X)\ast_A Y+X\ast_A (\huaE\cdot_A Y)-\huaE\cdot_A(X\ast_A Y),\quad \forall~X,Y\in\Gamma(A),\end{equation}
the multiplication $\cdot_\huaE$ is given by \eqref{eq:New associative mult} and $a_\huaE(X)=a_A(\huaE\cdot_A X)$.
\end{pro}

\begin{proof}
First, we show that $N$ is a Nijenhuis operator on the pre-Lie algebroid $(A,\ast_A,a_A)$. By \eqref{eq:pseudo-pre-HM1}, we have
 $$\Psi(\huaE\cdot_A X,\huaE,Y)=\Psi( Y,\huaE\cdot_AX,\huaE),\quad \forall~X,Y\in\Gamma(A),$$
 which implies that
 \begin{equation}\label{eq:pre-Lie Nij1}
   (\huaE\cdot_A X)\ast_A   (\huaE\cdot_A Y)=Y\ast_A (X\cdot_A \huaE\cdot_A \huaE)-(Y\ast_A (\huaE\cdot_A X))\cdot_A\huaE+((\huaE\cdot X)\ast_A Y)\cdot_A \huaE.
 \end{equation}
 Since $\huaE$ is  a pseudo-eventual identity on $A$, by \eqref{eq:Eventua31} and the symmetry of $\Psi$, we have
 $$\Psi(X,\huaE,Y)=-(\huaE\ast_A e)\cdot_A X\cdot_A Y.$$
 which implies that
 \begin{equation}\label{eq:pre-Lie Nij2}
   X\ast_A   (\huaE\cdot_A Y)=-(\huaE\ast_A e)\cdot_A X\cdot_A Y-(X\ast_A \huaE)\cdot_A Y-(X\ast_A Y)\cdot_A \huaE.
 \end{equation}
 By  \eqref{eq:Eventua31}, \eqref{eq:Eventua32}, \eqref{eq:pre-Lie Nij1}, \eqref{eq:pre-Lie Nij2} and the symmetry of $\Psi$, we have
 \begin{eqnarray*}
&&N(X)\ast_A N(Y)-N\big(N(X)\ast_A Y+X\ast_A N(Y)-N(X\ast_A Y)\big)\\
&=&(\huaE\cdot_A X)\ast_A(\huaE\cdot_A Y)-\huaE\cdot_A\big((\huaE\cdot_A X)\ast_A Y+X\ast_A (\huaE\cdot_A Y)-\huaE\cdot_A (X\ast_A Y)\big)\\
&=&Y\ast_A (X\cdot_A \huaE\cdot_A \huaE)-(Y\ast_A (X\cdot_A \huaE))\cdot_A \huaE-(X\ast_A \huaE)\cdot_A Y\cdot_A \huaE+(\huaE\ast_A e)\cdot_A X\cdot_A Y\cdot_A \huaE\\
  &=&Y\ast_A (X\cdot_A \huaE\cdot_A \huaE)-(Y\ast_A (X\cdot_A \huaE))\cdot_A \huaE-(Y\ast_A \huaE)\cdot_A X\cdot_A \huaE+(\huaE\ast_A e)\cdot_A X\cdot_A Y\cdot_A \huaE\\
  &=&\Psi(Y,X\cdot_A \huaE,\huaE)+(\huaE\ast_A e)\cdot_A X\cdot_A Y\cdot_A \huaE\\
  &=&-(\huaE\ast_A e)\cdot_A X\cdot_A Y\cdot_A \huaE+(\huaE\ast_A e)\cdot_A X\cdot_A Y\cdot_A \huaE\\
  &=&0.
\end{eqnarray*}
Thus $N=\huaE\cdot_A$ is a Nijenhuis operator on the pre-Lie algebroid $(A,\ast_A,a_A)$.

Also, $N=\huaE\cdot_A$ is a Nijenhuis operator on the commutative associative algebra $(\Gamma(A),\cdot_A)$. Therefore, $N=\huaE\cdot_A$ is a Nijenhuis operator on the \preF~$(A,\ast_A,\cdot_A,a_A)$. The second claim follows.
\end{proof}

\begin{cor}\label{cor:Nij-F-algebroid-pre}
Let $(M,\nabla,\bullet)$ be an $F$-manifold with a compatible flat connection and $\huaE$ a pseudo-eventual identity on $M$. Then there is a new \preF~ structure on $TM$ given by
\begin{eqnarray*}
X\bullet_\huaE Y&=&X\bullet Y\bullet \huaE,\\	
X\ast_\huaE Y &=&\nabla_{\huaE\bullet X}Y+\nabla_{\huaE\bullet Y}X-\huaE\bullet(\nabla_X Y),\\
a_\huaE(X)&=&\huaE\bullet X,\quad\forall~X,Y\in\frkX(M).
\end{eqnarray*}

\end{cor}

\subsection{Applications to integral systems}
\emptycomment{Let $(M,\nabla,\bullet)$ be an $F$-manifold with a compatible flat connection. Let $(u_1,u_2,\cdots,u_n)$ be the canonical coordinate systems on $M$. Denote by $$\frac{\partial}{\partial{u^i}}\bullet\frac{\partial}{\partial{u^j}}=c_{ij}^k\frac{\partial}{\partial{u^k}}.$$}

\begin{thm}{\rm(\cite{LPR11})}\label{thm:integra2}
 Let $(M,\nabla,\bullet)$ be an $F$-manifold with a compatible flat connection. Let $(u^1,u^2,\cdots,u^n)$ be the canonical coordinate systems on $M$. If $X$ and $Y$ in $\frkX(M)$ satisfy
 $$(\nabla_Z X)\bullet W= (\nabla_W X)\bullet Z,\quad (\nabla_Z Y)\bullet W= (\nabla_W Y)\bullet Z,\quad\forall~W,Z\in \frkX(M),$$
 then
 the associated flows
 \begin{equation}
   u_t^i=c_{jk}^iX^ku^i_x \quad \mbox{ and } \quad u_\tau^i=c_{jk}^iY^ku^j_x
 \end{equation}
 commute, where $\frac{\partial}{\partial{u^i}}\bullet \frac{\partial}{\partial{u^j}}=c_{ij}^k\frac{\partial}{\partial{u^k}}$, $X=X^i\frac{\partial}{\partial{u^i}}$ and $Y=Y^i\frac{\partial}{\partial{u^i}}$.
 \emptycomment{\item[\rm(ii)] the $(1,1)$-tensor field $(V_X)_j^i=c_{jk}^iX^k$ satisfies the condition
 $$\nabla_{\frac{\partial}{\partial{u^k}}}(V_X)_j^i=\nabla_{\frac{\partial}{\partial{u^j}}}(V_X)_k^i,$$
 which is the Hamiltonian theory of systems of hydrodynamic type \cite{DN89}.}

\end{thm}

\begin{pro}
 Let $(M,\nabla,\bullet)$ be an $F$-manifold with a compatible flat connection and an identity $e$. Assume that $\huaE_1,\huaE_2\in\frkX(M)$ are pseudo-eventual identities. Then the flows
 \begin{equation}
   u_t^i=c_{jk}^iX^ku^i_x, \quad  u_\tau^i=c_{jk}^iY^ku^j_x,\quad  u_s^i=X^{p}Y^qc_{jk}^ic_{pq}^ku^i_x
 \end{equation}
 commute, where $\frac{\partial}{\partial{u^i}}\bullet \frac{\partial}{\partial{u^j}}=c_{ij}^k\frac{\partial}{\partial{u^k}}$, $\huaE_1=X^i\frac{\partial}{\partial{u^i}}$ and $\huaE_2=Y^i\frac{\partial}{\partial{u^i}}$.
\end{pro}
\begin{proof}
Since $\huaE_1\in\frkX(M)$ and $\huaE_2\in\frkX(M)$ are pseudo-eventual identities on $(M,\nabla,\bullet)$, by Proposition \ref{lem:Eventual1}, $\huaE_1\bullet \huaE_2$ is also a pseudo-eventual identity. Thus $\huaE_1,\huaE_2$ and $\huaE_1\bullet \huaE_2$ satisfy \eqref{eq:Eventua32}. Furthermore, we have
$$\huaE_1\bullet \huaE_2=X^{p}Y^qc^k_{pq}\frac{\partial}{\partial{u^k}}.$$
By Theorem \ref{thm:integra2}, the claim follows.
\end{proof}

 \begin{thm}{\rm(\cite{LPR11})}\label{thm:interal 1}
 Let $(M,\nabla,\bullet)$ be an $F$-manifold with a compatible flat connection. Let $(u^1,u^2,\cdots,u^n)$ be the canonical coordinate systems on $M$ and $(X_{(1,0)},\cdots,X_{(n,0)})$ a basis of flat vector fields. Define the primary flows by
 \begin{equation}\label{eq:primary flow}
   u_{t_{(p,0)}}^i=c_{jk}^iX_{(p,0)}^ku^j_x.
 \end{equation}
 Then there is a well-defined  higher flows of the hierarchy defined as
 \begin{equation}\label{eq:principal hierarchy}
   u_{t_{(p,\alpha)}}^i=c_{jk}^iX_{(p,\alpha)}^ku^j_x.
 \end{equation}
 by means of the following recursive relations:
 \begin{equation}
   \nabla_{\frac{\partial}{\partial{u^j}}}X^i_{(p,\alpha)}=c_{jk}^iX_{(p,\alpha-1)}^ku^k_x.
 \end{equation}
 Furthermore, the flows of the principal hierarchy \eqref{eq:principal hierarchy} commute.
 \end{thm}

\begin{pro}
 Let $(M,\nabla,\bullet)$ be an $F$-manifold with a compatible flat connection and an identity $e$. Let $(X_{(1,0)},\cdots,X_{(n,0)})$ be a basis of flat vector fields. Assume that $\huaE\in\frkX(M)$ is a pseudo-eventual identity. Define the primary flows by
 \begin{equation}\label{eq:primary flow2}
   u_{t_{(p,0)}}^i=c_{jk}^mc_{ml}^i\huaE^lX_{(p,0)}^ku^j_x,
 \end{equation}
 where $\huaE=\huaE^i\frac{\partial}{\partial{u^i}}$. Then there is a well-defined  higher flows of the hierarchy defined as
 \begin{equation}\label{eq:principal hierarchy2}
   u_{t_{(p,\alpha)}}^i=c_{jk}^mc_{ml}^i\huaE^lX_{(p,\alpha)}^ku^j_x.
 \end{equation}
 by means of the following recursive relations:
 \begin{equation}
   \nabla_{\frac{\partial}{\partial{u^j}}}X^i_{(p,\alpha)}=c_{jk}^mc_{ml}^i\huaE^lX_{(p,\alpha-1)}^ku^k_x.
 \end{equation}
 Furthermore, the flows of the principal hierarchy \eqref{eq:principal hierarchy2} commute.
 \end{pro}
\begin{proof}
Since $\huaE\in\frkX(M)$ is a pseudo-eventual identity on $(M,\nabla,\bullet)$, by Proposition \ref{thm:construction F-algebroid3}, $(M,\nabla,\bullet_\huaE)$ is also an $F$-manifold with a compatible flat connection, where
$$X\bullet_\huaE Y=X\bullet Y\bullet \huaE,\quad \forall~X,Y\in\frkX(M).$$
Furthermore, we have $$\frac{\partial}{\partial{u^i}}\bullet_\huaE\frac{\partial}{\partial{u^j}}=c_{ij}^mc_{ml}^k\huaE^l\frac{\partial}{\partial{u^k}}.$$
By Theorem \ref{thm:interal 1}, the claim follows.
\end{proof}

 \end{document}